\documentclass[10pt,a4paper,onecolumn]{article} % a4paper

\usepackage[pdftex]{graphicx}
\usepackage{amssymb}
\usepackage{amsmath}
\usepackage{amsthm}
\usepackage{url}
\usepackage[font={small}, labelfont=bf]{caption} 
\usepackage{algorithm}
\usepackage{algpseudocode}
\usepackage{multirow}
\usepackage{graphicx}
\usepackage{booktabs}
\usepackage{bbm}
\usepackage{enumitem}
\usepackage{rotating}
\usepackage{subcaption}

\newcommand{\f}{\operatorname}

\usepackage[round]{natbib}

\newtheorem{theorem}{Theorem}[section]

\newtheorem{proposition}[theorem]{Proposition}
\newtheorem{definition}[theorem]{Definition}

\usepackage{hyperref}

\usepackage{xcolor} 
\usepackage{environ}         % provides \BODY
\usepackage{etoolbox}        % provides \ifdimcomp
\usepackage{graphicx}        % provides \resizebox
\newlength{\myl}
\let\origequation=\equation
\let\origendequation=\endequation
\RenewEnviron{equation}{
  \settowidth{\myl}{$\BODY$}  % calculate width and save as \myl
  \origequation
  \ifdimcomp{\the\linewidth}{>}{\the\myl}
  {\ensuremath{\BODY}}                             % True
  {\resizebox{\linewidth}{!}{\ensuremath{\BODY}}}  % False
  \origendequation
}

\title{Beyond the Power Law: Estimation, Goodness-of-Fit, and a Semiparametric Extension in Complex Networks}

\author{
Nixon Jerez-Lillo$^1$, Francisco A. Rodrigues$^2$, Paulo H. Ferreira$^3$, and Pedro L. Ramos$^1$\\\\
\normalsize{$^{1}$Departamento de Estadística, Pontificia Universidad Católica de Chile, Santiago, Chile}\\
\normalsize{$^{2}$Institute of Mathematical and Computer Sciences, University of S\~ao Paulo, S\~ao Carlos, Brazil}\\
\normalsize{$^{3}$Department of Statistics, Federal University of Bahia, Salvador, Brazil}
}

\date{}

\begin{document}

\maketitle

\begin{abstract}
Scale-free networks play a fundamental role in the study of complex networks and various applied fields due to their ability to model a wide range of real-world systems. A key characteristic of these networks is their degree distribution, which often follows a power-law distribution, where the probability mass function is proportional to $x^{-\alpha}$, with $\alpha$ typically ranging between $2 < \alpha < 3$. In this paper, we introduce Bayesian inference methods to obtain more accurate estimates than those obtained using traditional methods, which often yield biased estimates, and precise credible intervals. Through a simulation study, we demonstrate that our approach provides nearly unbiased estimates for the scaling parameter, enhancing the reliability of inferences. We also evaluate new goodness-of-fit tests to improve the effectiveness of the Kolmogorov-Smirnov test, commonly used for this purpose. Our findings show that the Watson test offers superior power while maintaining a controlled type I error rate, enabling us to better determine whether data adheres to a power-law distribution. Finally, we propose a piecewise extension of this model to provide greater flexibility, evaluating the estimation and its goodness-of-fit features as well. In the complex networks field, this extension allows us to model the full degree distribution, instead of just focusing on the tail, as is commonly done. We demonstrate the utility of these novel methods through applications to two real-world datasets, showcasing their practical relevance and potential to advance the analysis of power-law behavior.
\\

%Therefore, determining whether the degree distribution of an empirical network follows a power-law is a important concern. 

%, which can result in the rejection of the scale-free hypothesis

\textbf{Keywords}: Scale-free hypothesis, Power-law distribution, Complex networks, Degree distribution
\end{abstract}

\section{Introduction}\label{sec:introduction}

Galileo Galilei's observation is one of the first examples of a power law, which manifests in numerous natural and artificial systems. During his studies of medicine at the University of Pisa around 1580, he discovered a relationship between an animal's mass ($M$) and its size ($L$, length), represented as $M \approx L^3$. He also observed that the scaling of bone area ($S$) follows $S \approx L^2$ \citep{Thurner}. This finding indicates that mass increases at a faster rate than bone area, suggesting that larger animals are less common compared to smaller ones. 

\clearpage

Kepler's third law, Newton's law of gravitation, Coulomb's law of electromagnetism, and Stefan-Boltzmann law are examples of power laws, which can be expressed in the form $x^{-\alpha}$. In the Newton's and Coulomb's laws, $x$ represents the distance between planets or charges, where $\alpha=2$. Thus, the closer two planets or electric charges are, the stronger the force between them. Stokes’ Law, self-organized criticality, and the Curie-von Schweidler Law and Keibler’s Law, which relates metabolic rate to the mass of an organism, also present this relationship. The laws of Kepler, Newton, and Coulomb are deterministic models because they allow us to precisely predict the gravitational or electric force and the position of a planet at any time (see, for example, \citeauthor{moulton1952h}, \citeyear{moulton1952h}; \citeauthor{jackson2007electrodynamics}, \citeyear{jackson2007electrodynamics}; \citeauthor{ramaswamy2010mechanics}, \citeyear{ramaswamy2010mechanics}).

Due to the inherent stochastic elements in data, establishing the presence of power-law scaling presents a significant challenge. Let $X$ be a discrete random variable. $X$ has a power-law distribution if its probability mass function is given by:
\begin{equation}\label{eqn:pmf}
p(x)=\frac{x^{-\alpha}}{\zeta(\alpha,x_{\min})},\quad x \geq x_{\min} \geq 1,
\end{equation}
where $\alpha$ is the scaling exponent, $x_{\min}$ is the minimum value (lower bound) at which the power-law applies, and $\zeta(\alpha,x_{\min})$ is the Hurwitz zeta function defined by:
\begin{equation*}
\zeta(\alpha,x_{\min})= \sum_{k=0}^{+\infty} \left(k+ x_{\min}\right)^{-\alpha}.
\end{equation*}

In this context, a pivotal question arises concerning whether a given dataset adheres to a power-law distribution. This inquiry holds fundamental significance, as power laws can arise from diverse underlying mechanisms. For instance, power laws are linked to the fractal organization and can emerge through preferential attachment, self-organized criticality, and energy and cost minimization in information transmission \citep{Thurner}. Thus, identifying a power law in a system is crucial because it provides valuable insights into the underlying mechanisms governing the system's behavior. Furthermore, power laws are closely linked to scale invariance and self-similarity across different scales, providing a framework for prediction, modeling, and decision-making. Power laws also serve as indicators of critical points and phase transitions \citep{Stanley1971phase, Stauffer}, enabling the identification of thresholds and emergent phenomena. The recognition of power laws enhances our understanding of complex systems, with far-reaching implications in scientific, social, and economic domains~\citep{newman2005power, clauset2009power}. Lastly, \cite{ramos2021power, ramos2024power} introduced an extension of this model that incorporates change points to obtain a flexible piecewise model.  

Identifying power-law distributions is crucial when analyzing complex networks~\citep{newman2018networks}. Networks where the number of connections (called degree) follows a power-law distribution are known as scale-free networks~\citep{costa2007characterization, barabasi2009scale, barabasi2016network}. Graphically, power-law behavior is suggested by a roughly linear decline of points on a plot of the degree distribution when using a logarithmic scale on both axes, in contrast to exponential behavior. Let $ p(x) $ be the degree distribution of a network. If the histogram forms a straight line on $\log$-$\log$ scales, then $\log p(x) = -\alpha \log x + c$, where $\alpha$ is a constant that typically falls in the range $2 < \alpha < 3$ (scale-free hypothesis). Taking the exponential of both sides, this is equivalent to equation \eqref{eqn:pmf} with $c^{-1} = \log \zeta(\alpha, x_{\min})$.

Across most classes of networks, it is common to find evidence that the degree distribution of real-world networks follows a power-law distribution \citep{barabasi2016network}, such as the World Wide Web \citep{albert1999diameter}, biological networks \citep{prvzulj2007biological, khanin2006scale}, software systems \citep{wen2009software}, and citation distributions \citep{redner1998popular}. Identifying power-law behavior provides significant theoretical insights into its generative mechanism \citep{virkar2014power}. From a theoretical perspective, \cite{barabasi1999emergence} proposed a model capable of replicating the growth of numerous real networks. In this model, nodes are added one by one to the system with a probability proportional to the number of links existing nodes possess at that time. The networks generated by this model present a power-law distribution in the number of connections \citep{mata2020complex}.

The estimation of the power-law distribution from a dataset is typically conducted in two steps. The first step focuses on estimating the scaling parameter, which is generally determined using maximum likelihood approach (see \citeauthor{lehmann1999elements}, \citeyear{lehmann1999elements} for details), assuming that the lower bound of the distribution is fixed. The maximum likelihood estimator (MLE) for $\alpha$, denoted by $\hat\alpha_{\f{MLE}}$, is obtained by directly maximizing the logarithm of the likelihood function given by:
\begin{equation}\label{eqn:loglik}
\mathcal{L}(\alpha)=-n\log\zeta(\alpha,x_{\min})-\alpha\sum_{i=1}^{n}\log(x_i).
\end{equation}

The Fisher information element, which is commonly used to determine the standard error of the MLE, is obtained by calculating the negative second derivative of equation \eqref{eqn:loglik} and taking its expectation. It is given by:
\begin{equation}\label{eqn:fisher}
\mathcal{I}(\alpha)=n\left(\frac{\zeta''(\alpha,x_{\min})}{\zeta(\alpha,x_{\min})}-\left(\frac{\zeta'(\alpha,x_{\min})}{\zeta(\alpha,x_{\min})}\right)^2\right).
\end{equation}

Despite the MLE of $\alpha$ not having a closed-form expression, it is possible to obtain an approximate expression by using a continuity correction and neglecting quantities of order $x_{\min}^{-2}$. Thus, \cite{clauset2009power} define the approximate maximum likelihood estimator (AMLE) as:
\begin{equation}\label{eqn:amle}
\hat\alpha_{\f{AMLE}}= 1+n\left[\sum_{i=1}^{n}\log\left(\frac{x_i}{x_{\min}-\text{0,5}}\right)\right]^{-1},
\end{equation}
which is in fact similar to the MLE for the continuous power-law model.

The second step is focused on the estimation of $x_{\min}$. \cite{clauset2009power} discuss two approaches to this task. The first one was proposed by \cite{handcock2004likelihood} and consists in maximizing the Bayesian information criterion with respect to $x_{\min}$. This likelihood function is constructed by representing the data below $x_{\min}$ by separate probabilities, and above $x_{\min}$ by a power-law distribution. However, this method presents difficulties because it tends to underestimate it, leading to a biased estimation of the scaling parameter. On the other hand, \cite{clauset2007resolution} use the Kolmogorov-Smirnov (KS) statistic. This approach aims to find the value of $x_{\min}$ that makes the distribution of the data and the power-law model fit as closely as possible. In other words, it chooses $x_{\min}$ so that the cumulative distribution function (CDF) of the observed data, restricted to values greater than or equal to $x_{\min}$, is as similar as possible to the CDF of the power-law model fitted to the same range of data. \cite{clauset2009power} justifies that this method gives excellent results and performs better than the first one.

However, the literature indicates that MLEs are sensitive to outliers, which is critical for long-tailed distributions like the power-law distribution, as these outliers can disproportionately affect the estimation \citep{cordeiro2014introduction}. Moreover, empirical datasets often cover only a limited range of observations, making it challenging to clearly identify power-law behavior. Due to the importance of identifying this distribution in stochastic systems, a recent debate has suggested that power laws are not as common as has been reported \citep{broido2019scale}. Therefore, determining if a dataset follows a power law and inferring its parameters is a challenge for traditional statistical methods, making new approaches necessary.

In this paper, we develop Bayesian estimators for the scaling parameter of the power-law model using the Jeffreys prior. Our approach yields a proper posterior distribution and an efficient estimator in terms of bias and mean squared error. Additionally, \cite{kass1996selection} demonstrate that this prior leads to a posterior distribution with desirable properties: (i) invariance under one-to-one transformations, (ii) consistency under marginalization, and (iii) good coverage probabilities. Through a simulation study, we show that our proposed estimators provide nearly unbiased estimates for this parameter, even with small sample sizes, and improve upon the results obtained by the maximum likelihood method. Since the power-law model applies only above a minimum value, we discuss a piecewise extension that incorporates change points to enhance flexibility and model the entire dataset. This approach has been implemented for continuous power-law distributions by \cite{ramos2021power, ramos2024power}, achieving excellent results.

In addition, we address the problem of classifying networks as scale-free or, more broadly, determining whether data follow a power-law distribution. \cite{clauset2009power} proposed an approach based on the KS statistic to evaluate goodness-of-fit, where the approximate distributions of the test statistic are generated through a semiparametric bootstrap method. However, our findings show that this approach fails to maintain the nominal significance level, leading to incorrect discrimination of power-law samples for small and large sample sizes. To address this issue, we perform other well-known goodness-of-fit tests, such as the Cramér–von Mises, Watson, and Anderson–Darling tests. Furthermore, we assess the performance of these tests against various alternative hypotheses using Monte Carlo simulations. The simulation results reveal that the Watson test is the most powerful, consistently controlling the type I error at the specified nominal level, making it the preferred choice.

This paper is organized as follows: Section \ref{sec:bayes} introduces Bayesian estimators based on maximum a posteriori estimation with a Jeffreys prior and demonstrates that the resulting posterior distribution is proper. In Section \ref{sec:gof}, we review various goodness-of-fit tests and propose a novel joint test to correct the type I error of each individual test. Section \ref{sec:pwpl} introduces the piecewise extension of the discrete power-law model and discusses the estimation procedure. Section \ref{sec:simulation} presents four simulation studies designed to assess the procedures discussed, we say, the estimation and goodness-of-fit procedure for the standard and piecewise discrete power-law model. Section \ref{sec:application} analyzes two datasets: the corpus of networks examined by \cite{broido2019scale}, to assess the proposed methodology for the power-law model, and the word frequency distribution in the novel Moby Dick, to evaluate its piecewise extension. Our findings indicate that our proposal retains the power-law distribution for many networks where the KS test rejects it. Furthermore, we found that incorporating change-points into the power-law model results in a more flexible framework, effectively capturing entire datasets without truncating information. We conclude this study with a summary of the key findings and a discussion of their implications in Section \ref{sec:conclusion}.

%\clearpage
\section{Bayesian estimators}\label{sec:bayes}

Inferential procedures using Bayesian methods have become popular in recent decades. Under the Bayesian approach, parameters of a statistical model are treated as random variables, with a prior distribution representing beliefs about them. Central to Bayesian inference is Bayes' theorem, which updates these beliefs with observed data to form the posterior distribution. From the posterior distribution, various inferences can be made, including point estimation (e.g., the mean value of the posterior distribution, the posterior median that divides the distribution into two equal parts, and the posterior mode indicating the most likely value of the parameter), credible intervals (e.g., highest posterior density or equal-tailed), and hypothesis testing (see, for example,
\citeauthor{box2011bayesian}, \citeyear{box2011bayesian};
\citeauthor{gelman2013bayesian}, \citeyear{gelman2013bayesian}).

The choice of a prior distribution can be either informative or non-informative, depending on the available prior knowledge and the desired influence on the posterior. Non-informative priors, also known as diffuse or vague priors, are used when the goal is to let the data primarily drive the inferences rather than prior beliefs. These priors are designed to exert minimal influence on the posterior distribution. A common non-informative prior is the Jeffreys prior, which is widely used due to its invariance property under one-to-one transformations of parameters~\citep{jeffreys1946invariant}. This prior is proportional to the square root of the Fisher information element in equation \eqref{eqn:fisher}:
\begin{align*}
\pi_1\left(\alpha\right)\propto \sqrt{\frac{\zeta''(\alpha,x_{\min})}{\zeta(\alpha,x_{\min})}-\left(\frac{\zeta'(\alpha,x_{\min})}{\zeta(\alpha,x_{\min})}\right)^2}=\sqrt{\phi^{(2)}(\alpha,x_{\min})},
\end{align*}
where $\phi^{(i)}(\alpha,x_{\min})$ denotes the $i$-th derivative of $\log\zeta(\alpha,x_{\min})$.

The posterior distribution for $\alpha$ is given by:
\begin{align}\label{eqn:posterior1}
\pi_1(\alpha\mid\boldsymbol{x}) &= \frac{1}{d_1(\boldsymbol{x})}\frac{\sqrt{\phi^{(2)}(\alpha,x_{\min})}}{\left[\zeta(\alpha,x_{\min})\right]^{n}}\prod_{i=1}^{n}x_{i}^{-\alpha},
\end{align}
where:
\begin{align*}
d_1(\boldsymbol{x})=\int_{1}^{+\infty}\frac{\sqrt{\phi^{(2)}(\alpha,x_{\min})}}{\left[\zeta(\alpha,x_{\min})\right]^{n}}\prod_{i=1}^{n}x_{i}^{-\alpha}d\alpha.
\end{align*}

Due to the fact that the Jeffreys prior is an improper function, it is necessary to ensure that the posterior distribution is proper. Before moving forward, we need to introduce the following definitions and propositions. Let $\overline{\mathbb{R}} = \mathbb{R} \cup \{-\infty, \infty\}$ represent the extended real number line with the usual order $(\geq)$. Let $\mathbb{R}^{+}$ denote the positive real numbers, and $\mathbb{R}_0^{+}$ denote the non-negative real numbers, including zero. Similarly, we define $\overline{\mathbb{R}}^{+}$ and $\overline{\mathbb{R}}_0^{+}$ analogously.

\begin{definition}
Let $a \in \overline{\mathbb{R}}_0^{+}$ and $b \in \overline{\mathbb{R}}_0^{+}$. We say that $a \lesssim b$ if there exists $M \in \mathbb{R}^{+}$ such that $a \leq M \cdot b$. If $a \lesssim b$ and $b \lesssim a$, then we say that $a \propto b$.
\end{definition}

\begin{definition}
Let $g: \mathcal{U} \rightarrow \overline{\mathbb{R}}_0^{+}$ and $h: \mathcal{U} \rightarrow \overline{\mathbb{R}}_0^{+}$, where $\mathcal{U} \subset \mathbb{R}$. We say that $g(x) \lesssim h(x)$ if there exists $M \in \mathbb{R}^{+}$ such that $g(x) \leq M \cdot h(x)$, for every $x \in \mathcal{U}$. Moreover, if $g(x) \lesssim h(x)$ and $h(x) \lesssim g(x)$, then we denote this by $g(x) \propto h(x)$.
\end{definition}

\begin{proposition}\label{prop0}
Let $g:(a, b) \rightarrow \mathbb{R}^{+}$ and $h:(a, b) \rightarrow \mathbb{R}^{+}$ be continuous functions on $(a, b) \subset \mathbb{R}$, where $a, b \in \overline{\mathbb{R}}$. Then, $g(x) \propto g(x)$ if and only if $\lim\limits_{x \to a} {g(x)}/{h(x)} < +\infty$ and $\lim\limits_{x \to b} {g(x)}/{h(x)} < +\infty$.
\end{proposition}

\begin{proof}
A detailed proof of this proposition can be found in \cite{ramos2023power}.
\end{proof}

\begin{proposition}\label{prop1} For all $r\geq 0$ and $n\geq 1$, we have that:
\begin{equation*}
\int_{1}^{+\infty}\alpha^r\frac{\sqrt{\phi^{(2)}(\alpha,x_{\min})}}{\left[\zeta(\alpha,x_{\min})\right]^{n}}\prod_{i=1}^{n}x_{i}^{-\alpha} d\alpha<+\infty.
\end{equation*}
\end{proposition}

\begin{proof}
The detailed proof is provided in Appendix \ref{appendix:proof1}.
\end{proof}

Using this result, we have proven that the posterior is proper, ensuring the reliability of the inferences drawn from it. Additionally, we have demonstrated that all posterior moments are finite. This finiteness supports the use of numerical techniques, such as Markov chain Monte Carlo methods \citep{gamerman2006markov}, which depend on the proper normalization of the posterior distribution. However, it is important to note that these techniques will not be utilized in this work.

Although most mathematical and statistical softwares do not have the derivatives of the Hurwitz zeta function implemented, we can consider its relationship with the zeta function to calculate $\phi^{(2)}(\alpha,x_{\min})$ for inference purposes, given the posterior distribution in \eqref{eqn:posterior1}. The first two derivatives of the Hurwitz zeta function can be expressed as (see Appendix \ref{appendix:derivatives_HZfunction}):
\begin{align*}
    \zeta'(\alpha,x_{\min}) &= \zeta'(\alpha) + \sum_{k=1}^{x_{\min}-1} \frac{\log(k)}{k^\alpha},\\
    \zeta''(\alpha,x_{\min}) &= \zeta''(\alpha) - \sum_{k=1}^{x_{\min}-1} \frac{\left[\log(k)\right]^2}{k^\alpha},
\end{align*}
where $\zeta'(\alpha)$ and $\zeta''(\alpha)$ denote, respectively, the first and second derivatives of the Riemann zeta function.

In this analysis, we will use the maximum a posteriori (MAP) estimate as our point estimation method for the scaling parameter. To obtain this estimate, we need to maximize the equation \eqref{eqn:posterior1}, or equivalently, its logarithm. The MAP estimator for $\alpha$, denoted by $\hat{\alpha}_{\text{MAP}}$, is found as the root of the following equation: 
\begin{align}\label{eqn:map1}
\dfrac{d\log\left(\phi^{(2)}(\alpha,x_{\min})\right)}{d\alpha} \Bigg|_{\alpha = \hat\alpha_{\f{MAP}}}-2n\,\dfrac{\zeta^\prime(\hat\alpha_{\f{MAP}},x_{\min})}{\zeta(\hat\alpha_{\f{MAP}},x_{\min})}-2\sum_{i=1}^{n}\log(x_i) = 0.
\end{align}

Due that solving directly the equation above is not an easy task, specially because it is difficult to obtain the second derivative of $\log\left(\phi^{(2)}(\alpha,x_{\min})\right)$, we can use the Jeffreys prior derived from the continuous power-law model \citep{clauset2009power}, which is given by:
\begin{align*}
    \pi_2(\alpha) \propto \dfrac{1}{\alpha - 1},
\end{align*}
which eliminates the dependence on the lower bound in their formulation. This function provides a good approximation of the Jeffreys prior in the discrete model (see Figure \ref{fig1}), thus avoiding computational difficulties and ensuring a more tractable estimation process.

\begin{figure}[!ht]
    \centering
    \includegraphics[scale=0.6]{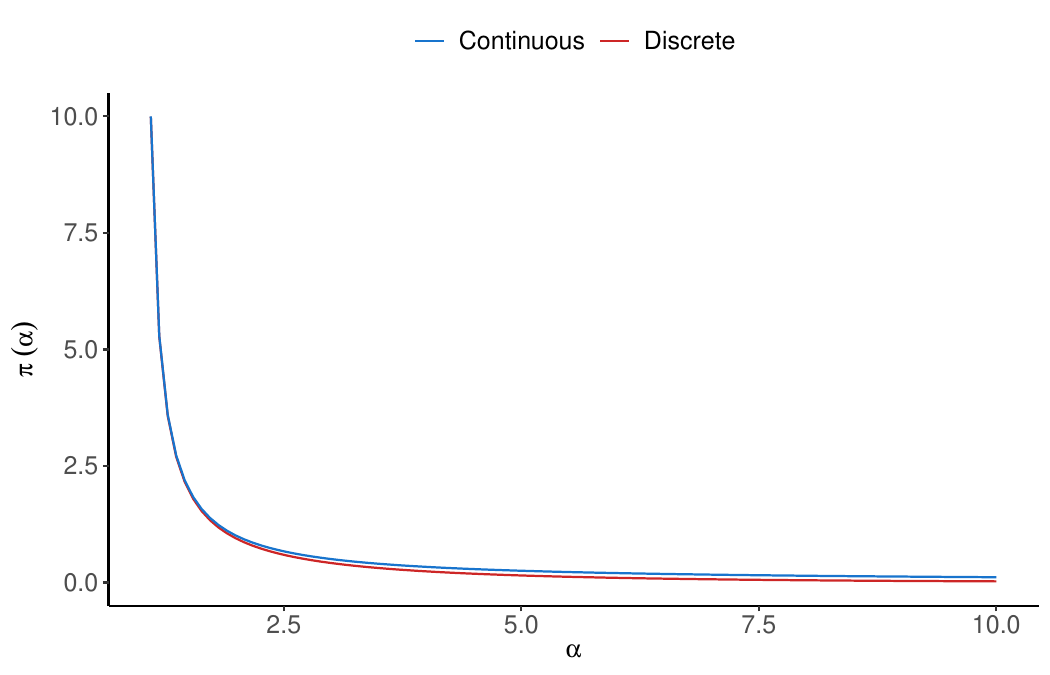}
    \caption{Jeffreys priors for the discrete and continuous power-law model taking $x_{\min}=1$.}\label{fig1}
\end{figure}
	
The posterior distribution, derived from the Jeffreys prior of the continuous model, is given by:
\begin{equation}\label{eqn:posterior2}
\begin{aligned}
\pi_2(\alpha\mid\boldsymbol{x}) &= \frac{1}{d_2(\boldsymbol{x})}\frac{(\alpha-1)^{-1}}{\left[\zeta(\alpha,x_{\min})\right]^{n}}\prod_{i=1}^{n}x_{i}^{-\alpha},
\end{aligned}
\end{equation}
where: 
\begin{equation*}
d_2(\boldsymbol{x})=\int_{1}^{+\infty}\frac{(\alpha-1)^{-1}}{\left[\zeta(\alpha,x_{\min})\right]^{n}}\prod_{i=1}^{n}x_{i}^{-\alpha}d\alpha.
\end{equation*}

Again, we will need to ensure that this new posterior distribution is proper.

\begin{proposition}\label{prop2} For all $r\geq 0$ and $n\geq 1$, we have that:
\begin{equation*}
\int_{1}^{+\infty}\alpha^r\frac{(\alpha-1)^{-1}}{\left[\zeta(\alpha,x_{\min})\right]^{n}}\prod_{i=1}^{n}x_{i}^{-\alpha} d\alpha<+\infty.
\end{equation*}
\end{proposition}

\begin{proof}
The proof of this proposition is provided in detail in Appendix \ref{appendix:proof2}.
\end{proof}

From the previous result, we have demonstrated not only that the posterior is finite but also that all its moments are finite, specifically, its mean and variance. Consequently, the posterior is well-behaved and can be used for further inference. 

The MAP estimate based on Jeffreys prior from continuous power-law model is obtained as the root of the following equation:
\begin{align}\label{eqn:map2}
\frac{1}{\hat{\alpha}_{\text{MAP}} - 1} +
2n \frac{\zeta^\prime(\hat{\alpha}_{\text{MAP}}, x_{\min})}{\zeta(\hat{\alpha}_{\text{MAP}}, x_{\min})} +
2 \sum_{i=1}^{n} \log(x_i) = 0,
\end{align}
which is derived by taking the first derivative of the logarithm of the posterior distribution, as expressed in equation \eqref{eqn:posterior2}, and setting it equal to zero.

\section{Testing power-law distribution}\label{sec:gof}

After estimating the parameters of a statistical model, the next step is to evaluate how well it fits the data. Indeed, when we propose a parametric statistical model, we assume that the underlying mechanism generating the data, described by a probability distribution, can be recovered by evaluating a specific value within the parametric space of these families of probability distributions, which is suggested by a particular estimation method \citep{mccullagh2002statistical}. However, in practice, this hypothesis needs to be tested, or at least checked, to determine if there is evidence to reject it.

A goodness-of-fit (GOF) test is a statistical test used to determine how well a model fits a set of observed data. These tests are based on measuring the ``distance" between the distribution of the empirical data and the hypothesized model. Let $p_0(x)$ denote the probability mass function from which the data were generated. We are interested in testing the following hypothesis:
\begin{align*}
    H_0: p_0(x) \in \mathcal{P} \quad \text{ versus } \quad H_1: p_0(x) \notin \mathcal{P},
\end{align*}
where $\mathcal{P}$ refers to the collection of all probability mass functions defined as \eqref{eqn:pmf}.

Let $F_n(x)$ denote the empirical distribution function of the random sample and $F(\cdot\,; \hat{\boldsymbol\theta})$ denote the CDF of the fitted model in $\mathcal{P}$. In this work, we have focused on four widely used tests: the Kolmogorov-Smirnov, the Cram\'er-von Mises, the Watson, and the Anderson-Darling tests (see \citeauthor{d2017goodness}, \citeyear{d2017goodness}; \citeauthor{zamanzade2019edf}, \citeyear{zamanzade2019edf} for a detailed review). 

\subsection{Kolmogorov-Smirnov test}

The KS test is one of the most popular GOF statistical tests used to measure how well a random sample agrees with a specific distribution. Its statistic is defined as:
\begin{align*}
    K = \max |F_n(x)-F(x; \hat{\boldsymbol{\theta}})|.
\end{align*}

The distribution of $K$ under $H_0$ is intractable because its theoretical properties depend on various factors, such as the distribution being tested, the method of estimation, the values of parameter estimates, and the sample size.  In practice, the parametric bootstrap is used to approximate its distribution. This approach involves repeatedly resampling from an estimated model and recalculating the test statistic, thus providing an empirical distribution to approximate the critical values for the KS test.

On the other hand, \cite{clauset2009power} describe a semiparametric bootstrap method to obtain synthetic data. Let $n_{\text{tail}}$ be the number of observations greater than $x_{\min}$. The method is as follows: with probability $n_{\text{tail}}/n$, each observation is drawn from a power-law distribution with scaling parameter $\hat{\alpha}$ for sample values greater than $x_{\min}$. Otherwise, with probability $1 - n_{\text{tail}}/n$, we select a value uniformly at random from the elements of the observed dataset less than $x_{\min}$. After that, we must fit each synthetic dataset individually to its own power-law model and calculate the KS statistic for each one relative to its own model. Finally, we define the $p$-value as the fraction of times that the resulting statistic is larger than the value for the empirical data. If the $p$-value is less than the chosen significance level, the null hypothesis that the data follow the fitted model is rejected; otherwise, it is not rejected.

%Additionally, the authors recommend generating 2,500 synthetic samples to achieve accuracy to 2 decimal places of the $p$-value, based on an analysis of the expected worst-case performance of the test.

\subsection{Alternative goodness-of-fit tests}

Some alternatives to the KS statistic include the Cramér-von Mises \((W^2)\), Watson \((U^2)\), and Anderson-Darling \((A^2)\) statistics. Each of them differ in how they evaluate the discrepancies between the empirical CDF and the hypothesized one. For instance, the Cramér-von Mises statistic evaluates the average squared differences across the entire distribution, the Watson statistic adjusts these differences to eliminate effects of location, and the Anderson-Darling statistic places greater emphasis on discrepancies in the tails of the distribution. The corresponding statistics are defined as follows:
\begin{align*}
    W^{2} & =\frac{1}{12 n}+\sum_{i=1}^n\left(\hat{U}_{(i)}-\frac{2 i-1}{2 n}\right)^2, \\
    U^{2} & =W^{2}-n\left(\overline{\hat{U}}-1 / 2\right)^2, \\
    A^{2} & =-n-\frac{1}{n} \sum_{i=1}^n(2 i-1)\left[\log \hat{U}_{(i)}+\log \left(1-\hat{U}_{(n+1-i)}\right)\right],
\end{align*}
where $\hat{U}_i=F(X_i ; \hat{\boldsymbol{\theta}})$ denotes the fitted CDF, $\hat{U}_{(1)} \leq \ldots \leq \hat{U}_{(n)}$ denote their ordered statistics, $\overline{\hat{U}}=\sum_{i=1}^n \hat{U}_i / n$, and $X_i, i=$ $1, \ldots, n$, are the given data points. If $F(\cdot\,; \hat{\theta})$ provides a close approximation to the true distribution $F_0$, then $\hat{U}_i$ will be asymptotically uniformly distributed on the interval $[0,1]$. These formulas indicate that all test statistics are designed to detect departures of $\hat{U}_i$ from uniformity on $[0,1]$. Therefore, for each statistic, we reject the null hypothesis if the statistic value is larger than the $100(1-\gamma) \%$ quantile of the corresponding null distribution.

\section{Discrete piecewise power-law model}\label{sec:pwpl}

As mentioned in the Introduction section, the power-law model is useful for describing the tails of many empirical probability distributions. However, it often fails to accurately capture the entire distribution. One way to address this issue is to extend the model to a piecewise model. A piecewise distribution refers to a model that incorporates change points, which divide the data into intervals and allow for the application of different power-law models to each interval, producing a more flexible model than the original.

Consider a partition over $\mathcal{R} = [\tau_{(0)}, +\infty)$ yielding $\mathcal{R} = \mathcal{R}_1 \cup \ldots \cup \,\mathcal{R}_{k + 1}$, with $\mathcal{R}_{j} = [\tau_{(j-1)}, \tau_{(j)})$ for $j = 1, \dots, k + 1$, where $\tau_{(0)} \geq 1$ is the sample minimum, $\tau_{(1)}< \ldots< \tau_{(k)}$ are change points, and $\tau_{(k+1)} = +\infty$. The probability mass function of the discrete piecewise power-law model is given by:
\begin{align}\label{eqn:pmf_pw}
p(x)&=\sum_{j=1}^{k + 1} \left[\frac{x^{-\alpha_j}}{\zeta(\alpha_j,\tau_{(j-1)})} \cdot C_{j-1} \right] \mathbbm{1}_{\mathcal{R}_j}(x),
\end{align}
where $C_0=1$, and $C_{i}=\displaystyle\prod_{h=1}^{i} \frac{\zeta(\alpha_h,\tau_{(h)})}{\zeta(\alpha_{h},\tau_{(h-1)})},$ with $i=1,\ldots,k$, normalization constants. 

\begin{proposition}\label{prop3}
Let $X$ be a random variable with the probability mass function in \eqref{eqn:pmf_pw}. The variable $X \mid X \geq \tau_{(k)}$ follows a power-law distribution with a lower bound of $\tau_{(k)}$ and a scaling parameter of $\alpha_{k+1}$.
\end{proposition}

\begin{proof}
The full proof is included in Appendix \ref{appendix:proof3}.
\end{proof}

Estimating the parameters of this model presents certain challenges because the presence of change points causes the model to fail to meet regularity conditions \citep{casella2021statistical}. To address this issue, we propose a brute-force search procedure inspired by \cite{xu2024pwexp}. First, we assume that the change points are located within the observations due to their discrete nature. Then, for a specific set of $k$ fixed change points, we define the log-likelihood function as follows:
\begin{align}\label{eqn:loglik_pw}
\mathcal{L}(\boldsymbol{\theta})&=\sum_{j=1}^{k + 1}\left[n_j\Big(\log C_{j-1} - \log \zeta(\alpha_j,\tau_{(j)})\Big) -\alpha_j\sum_{i:x_i\in \mathcal{R}_{j}} \log x_i\right],
\end{align}
where $\boldsymbol{\theta} = (\alpha_1, \ldots, \alpha_{k+1})$ contains all the scaling parameters for each segment induced by the partition. We proceed by estimating all the possible models with $k$ change points by maximizing the log-likelihood function in \eqref{eqn:loglik_pw}. Finally, through an exhaustive search, we select the model that achieves the highest likelihood among all the possible configurations, according to the maximum likelihood principle.

Finally, the optimal number of change points is determined by selecting the model that offers the best performance under any conventional criteria. In this step, we can define as many piecewise models as the number of change points we want to evaluate. After estimating their locations and scaling parameters, we can use the Akaike Information Criterion (AIC), defined as $\text{AIC} = 2p - 2\mathcal{L}(\hat{\boldsymbol{\theta}})$, where $p$ is the number of estimated parameters, or the Bayesian Information Criterion (BIC), given by $\text{BIC} = p\log(n) - 2\mathcal{L}(\hat{\boldsymbol{\theta}})$, to adhere to the principle of parsimony \citep{zhang1992distributional}.

To further quantify the uncertainty in the estimates of the scaling parameter for each partition, we can employ two well-known methods: the asymptotic distribution of the MLE and bootstrap resampling. The first method involves using the Hessian matrix, $ \mathbf{H} $, of the log-likelihood function, evaluated at its maximum, which, as shown in \cite{lehmann1999elements}, can be used to estimate the covariance matrix of the parameters, as it is a consistent estimator of the Fisher information. Under regularity conditions, the asymptotic confidence interval for $ \alpha_j $ at the level $ 1 - \gamma $ is given by:
\begin{align*}
\hat{\alpha}_j \pm z_{\gamma/2} \cdot \sqrt{-[\mathbf{H}^{-1}]_{jj}},\quad j=1,\ldots, k+1,
\end{align*}
where $ z_{\gamma/2} $ is the critical value from the standard normal distribution.

Alternatively, we may use a second method based on the bootstrap technique \citep{efron1982jackknife}, which consists of generating $B$ resampled datasets by sampling with replacement from the original dataset. For each bootstrap sample, we fit the model and obtain an estimate $ \hat{\alpha}^{(i)} $, $i=1, \ldots, B$. The distribution of these estimates, known as the bootstrap distribution, allows us to estimate the bias of our estimators and construct confidence intervals. By using bootstrap to assess the bias of the estimator, we can define the bootstrap bias-corrected estimator (see \citeauthor{ferrari1998bootstrap}, \citeyear{ferrari1998bootstrap} for details) as follows:
\begin{align*}
    \hat\alpha_j^{\text{(boot)}} = 2\hat{\alpha}_j - \frac{1}{B} \sum_{i=1}^{B} \hat{\alpha}_j^{(i)}, \quad j=1,\ldots, k+1,
\end{align*}
where $ \hat{\alpha}_j $ is the estimate from the original data. On the other hand, the confidence interval is derived from the empirical quantiles of the bootstrap distribution, we say, $\left[ \hat{\alpha}_{\gamma/2}^{*}, \hat{\alpha}_{1-\gamma/2}^{*} \right],$ where $ \hat{\alpha}_{\gamma/2}^{*} $ and $ \hat{\alpha}_{1-\gamma/2}^{*} $ are the $ \gamma/2 $-th and $ 1 - \gamma/2 $-th percentiles of the bootstrap distribution, respectively.

\clearpage
\section{Simulation study}\label{sec:simulation}

In this section, we present a comprehensive simulation study designed to evaluate the effectiveness and performance of the different estimation and GOF procedures introduced in this work. To achieve this, we have divided the analysis into three separate studies.

\subsection{Study 1}\label{sec:sim_1}

In this initial step, we focus on assessing the performance of the estimators for the scaling parameter of a power-law distribution: the MLE (derived from the maximization of the equation \eqref{eqn:loglik}), the AMLE (defined in equation \eqref{eqn:amle}), the MAP estimator with Jeffreys prior obtained for the discrete power-law model (derived from equation \eqref{eqn:map1}), and the MAP estimator with Jeffreys prior obtained for the continuous power-law model (derived from equation \eqref{eqn:map2}). These estimators were evaluated under four scenarios, each with distinct parameter values. Case 1 involved $x_{\min} = 4$ and $\alpha = 4$. Case 2 involved $x_{\min} = 6$ and $\alpha = 4$. Case 3 involved $x_{\min} = 4$ and $\alpha = 6$. Case 4 involved $x_{\min} = 6$ and $\alpha = 6$. The simulation study employed 5,000 synthetic samples for each case, with sizes $n=10, 20, \ldots, 100$. These samples were generated using the \texttt{rpldis} function from the \texttt{poweRlaw} package \citep{powerRlaw} in the \texttt{R} software.

We consider three key measures to evaluate the performance of the estimation procedures: bias, mean squared error (MSE), and 95\% coverage probability (CP) (see \citeauthor{morris2019using}, \citeyear{morris2019using} for a mathematical definition of these measures). Bias measures the difference between the expected value of an estimator and the true parameter value, providing insight into the estimator's accuracy. MSE, calculated as the average of the squared differences between the estimated and true values, reflects both the variance and bias of the estimator. Finally, CP assesses the proportion of times the true parameter value lies within the specified confidence intervals for the MLEs and AMLEs (obtained using their asymptotic distribution) or within the credible intervals of the posterior distributions in \eqref{eqn:posterior1} and \eqref{eqn:posterior2} (obtained using the 2.5\% and 97.5\% percentiles of each distribution, respectively). Using these criteria, we are searching for an estimator that yields a bias closer to zero, smaller MSEs, and CPs near 0.95.

Figure \ref{fig2} shows the biases, MSEs, and CPs from the estimates of $\alpha$ obtained using the four estimation methods described earlier. The horizontal lines in the figures correspond to the expected values for each metric: zero for bias and MSEs, and 0.95 for CP.

As shown in these figures, the most precise estimators in terms of bias across all cases are the Bayesian estimators, particularly the one derived from the posterior distribution in \eqref{eqn:posterior1}, which is nearly unbiased even with samples as small as size 10. In contrast, the frequentist estimators perform worse, with the AMLE showing a notable bias that increased as the sample size increased. Regarding MSE, the MLE exhibits the poorest performance, followed by the Bayesian estimators, while the AMLE performs the best, indicating that its large bias is compensated for by its small variability. Finally, the CPs are close to 0.95 in all cases except for the AMLE, which is affected by its large bias and reduced variability. 

Overall, these results suggest that the MAP estimator derived from equation \eqref{eqn:map1} is preferable for estimating the scale parameter of the power-law distribution due to its balanced performance across all evaluation metrics. Its near-unbiased nature, even for small sample sizes, competitive MSE, and reliable CPs make it a robust choice compared to the frequentist alternatives.

\begin{figure}[H]
    \centering
    \includegraphics[scale=0.6]{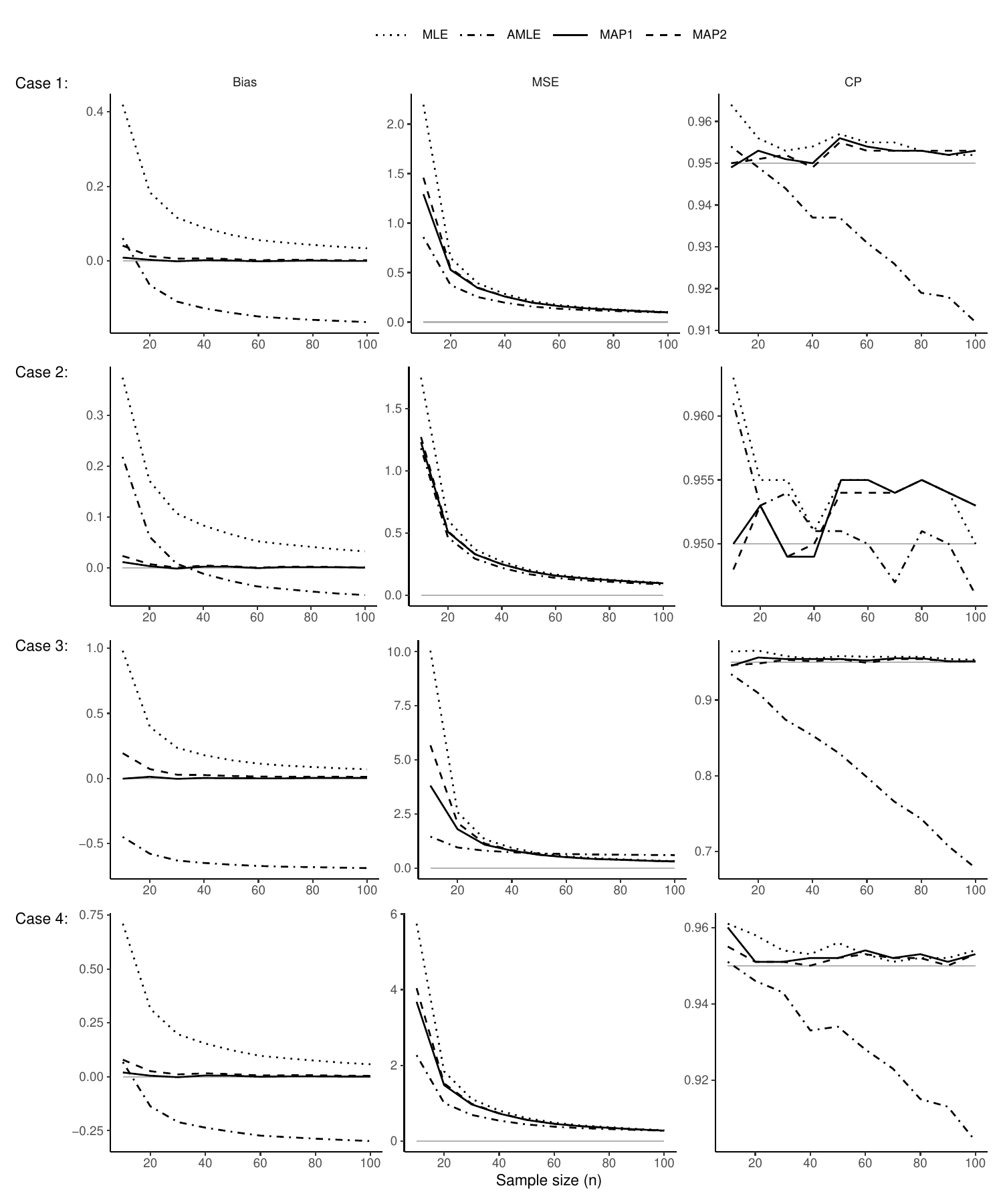}
    \caption{Performance of the estimators in terms of bias (left), MSE (middle), and CP (right) across different cases using random samples of sizes $n=10, 20, \ldots, 100$ generated from a power-law distribution.}\label{fig2}
\end{figure}

\subsection{Study 2}\label{sec:sim_2}

In this section, we examine the performance of the introduced GOF test in Section \ref{sec:gof} through Monte Carlo simulations. Two key aspects of any statistical test are its rejection rate when the null hypothesis $H_0$ is true (known as the size) and its rejection rate when $H_0$ is false (known as the power). An effective test is one that has a size close to the nominal significance level and demonstrates high power.

All statistics were calculated using the \texttt{dgof} package \citep{dgof} in the \texttt{R} software, except $K$ statistics, which was computed using the code available in \cite{clauset_powerlaws}, which implements the semiparametric bootstrap method described in the previous section, with a minor modification to include the scaling parameter as a function argument. 

\paragraph{Size of the test.}

We estimate the empirical sizes of all tests at the 5\% and 10\% significance levels. Consequently, we do not reject the power-law distribution for samples where the \(p\)-values exceed 0.05 and 0.1, respectively. In this initial stage, we used 1,000 pseudo-random samples of sizes $n$ = 100, 200, 300, 400, and 500. Figures \ref{fig3} and \ref{fig10} in Appendix \ref{appendix:study2} present the empirical sizes for all GOF tests in both significance levels.

As can be seen, the $W^2$, $U^2$, and $A^2$ test statistics maintain the nominal significance level and correctly discriminate between samples for small and moderate sample sizes. In contrast, the KS test rejects around 40\% of power-law samples across cases. This behavior violates a well-known theorem, which states that under the null hypothesis, the $p$-values of a test should follow a uniform distribution, meaning that the rejection rate of these samples should match the significance level.
  
\paragraph{Power of the test.} 

In this part of the study, we will conduct a simulation study to assess the power of the same GOF tests discussed previously under various alternative distributions. Specifically, we analyze the rejection rates of all tests against the four alternative hypotheses at a 5\% and 10\% significance levels. Case 1 involved a discrete Exponential distribution with $x_{\min} = 2$ and $\lambda = 0.4$. Case 2 involved a discrete Exponential distribution as well with $x_{\min} = 2$ and $\lambda = 0.6$. Case 3 involved a discrete Poisson distribution with $x_{\min} = 4$ and $\mu = 6$. Case 4 involved a discrete Poisson distribution as well with $x_{\min} = 4$ and $\mu = 4$. These alternative distributions are introduced in detail in Appendix D of \cite{clauset2009power}. Simulations were carried out again with 1,000 pseudo-random samples of sizes $n = 100, 200, 300, 400$, and 500. Figure \ref{fig4} and \ref{fig11} in Appendix \ref{appendix:study2} shows the empirical power for all GOF tests in both cases.

We observed that the KS test is the most powerful among the tests considered, exhibiting a high rejection rate. This high power allows it to detect deviations from the null hypothesis more effectively than the other tests. On the other hand, the Watson test is the second most powerful, while the Cramér-von Mises and Anderson-Darling tests are the least powerful, showing similar performance.

\begin{figure}[!ht]
    \centering
    \includegraphics[scale = 0.4]{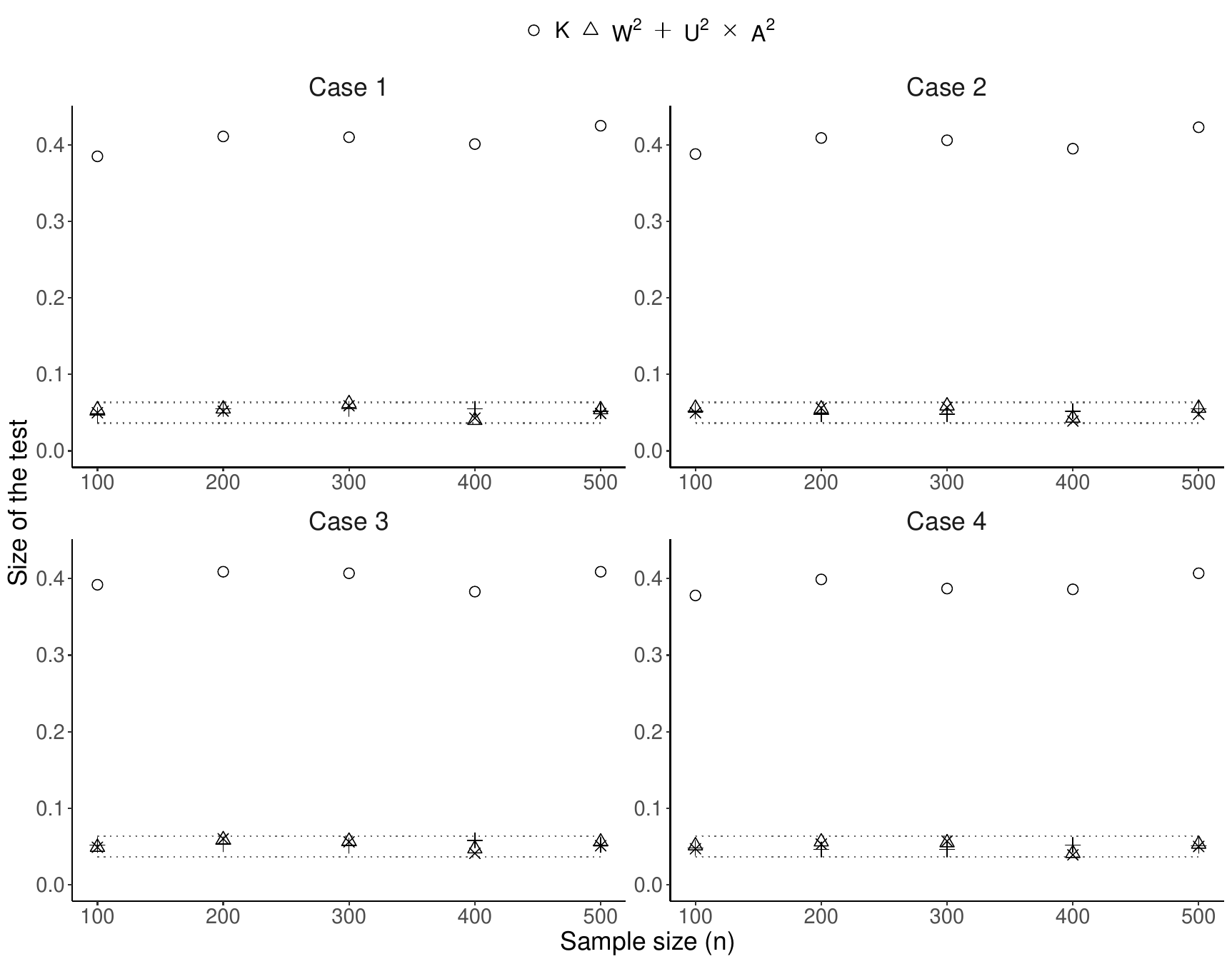}
    \caption{Empirical sizes of the GOF tests based on the $K$, $W^2$, $U^2$, and $A^2$ statistics at a 5\% significance level. The dashed line represents the limits of $0.05 \pm 1.96 \sqrt{0.05(1-0.05)/\text{1,000}}$, derived from the normal approximation.}\label{fig3}
\end{figure}

\begin{figure}[!ht]
    \centering
    \includegraphics[scale=0.4]{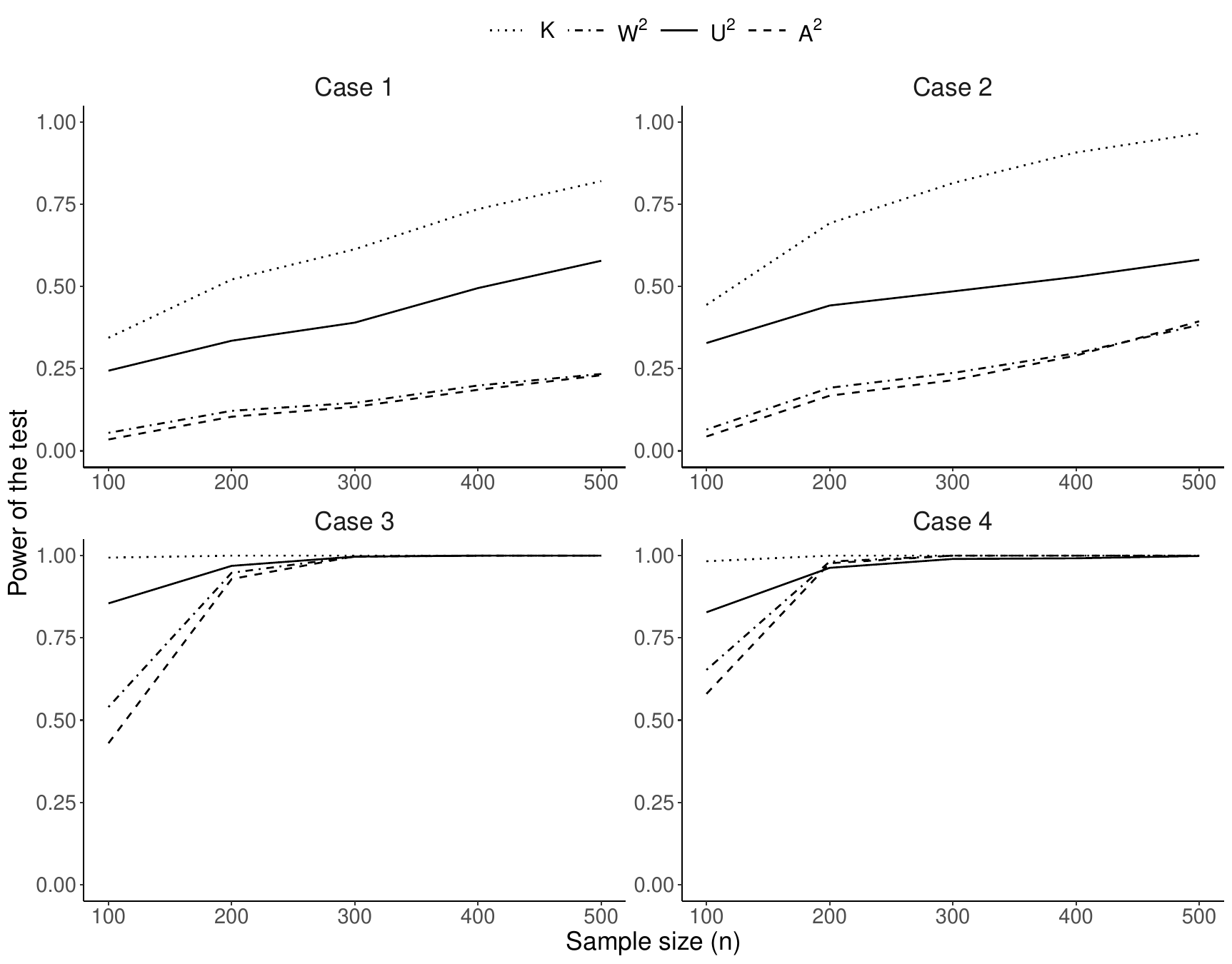}
    \caption{Empirical power of GOF tests based on the $K$, $W^2$, $U^2$, and $A^2$ statistics at a 5\% significance level.}\label{fig4}
\end{figure}

\paragraph{Conclusion.} The KS test is the most powerful test, as shown by its performance across different scenarios and sample sizes, but it cannot control the type I error rate, meaning it rejects the null hypothesis too often, even when it is true, which makes it unreliable. This overreaction makes the test seem more effective than it actually is. In contrast, the Watson test performs well in terms of power while maintaining a controlled type I error rate, so it is highly recommended for use in practical situations.

\clearpage
\subsection{Study 3}\label{sec:sim_3}

In our third simulation study, we focus on evaluating the performance of the estimation methods for the discrete piecewise power-law model. Specifically, we aim to assess the precision of two estimators of the scaling parameter for each partition: the MLE and the bootstrap bias-corrected MLE. For this task, we used the same metrics as in Study 1: bias and MSE. Additionally, we are interested in generating confidence intervals for these parameters by using both the asymptotic distribution of the MLE and bootstrap simulations. The 95\% CP was employed to compare the performance of both methods. Initially, we employed 100 bootstrap samples due to the computational cost. Furthermore, we evaluated the accuracy of the change-point estimation, calculated as the percentage of times the estimated change points coincided with the specified ones. Due to the computational cost of evaluating all possible change-point combinations, we constrained the search to the 90th percentile of each sample. This decision was motivated by the expectation that change points are more likely to occur where the data are concentrated, rather than in the tail of extreme events, which have a lower probability of occurrence.

Each of these aspects was evaluated under four distinct scenarios: Case 1 involved $\tau_{(1)} = 5$ and $\boldsymbol{\theta} = (1.2, 3.0)$; Case 2 involved $\tau_{(1)} = 5$ and $\boldsymbol{\theta} = (1.5, 3.5)$; Case 3 involved $(\tau_{(1)}, \tau_{(2)}) = (3, 5)$ and $\boldsymbol{\theta} = (1.5, 1.8, 3.0)$; and Case 4 involved $(\tau_{(1)}, \tau_{(2)}) = (3, 5)$ and $\boldsymbol{\theta} = (1.5, 1.2, 3.5)$. The simulation study was based on 1,000 replicates for each case, with sample sizes of $n=100, 250, 500, 750,$ and 1,000. These samples were generated using the inverse transformation method dealing on Appendix \ref{sec:samples}.

\begin{table}[!ht]
\centering
\caption{Bias, MSE, and CP for the scaling parameter estimators computed via bootstrap, along with the accuracy (Accu.) of the change-point estimates for the discrete piecewise power-law model.}
\label{tab1}
\scalebox{0.85}{
\rotatebox{0}{
\begin{tabular}{ccccccccccccccccccc}
  \hline
  & & \multicolumn{3} {c} { $\hat\alpha_1^{(\text{boot})}$ } & & $\hat\tau_{(1)}$ & & \multicolumn{3} {c} { $\hat\alpha_2^{(\text{boot})}$ } & & $\hat\tau_{(2)}$ & & \multicolumn{3} {c} { $\hat\alpha_3^{(\text{boot})}$ } \\
  \cline{3-5} \cline{7-7} \cline{9-11} \cline{13-13} \cline{15-17}
  & $n$ & Bias & MSE & CP & & Accu. & & Bias & MSE & CP & & Accu. & & Bias & MSE & CP \\ 
  \hline
   \multirow{1}{*}{Case 1:} & 100 & 0.00 & 0.00 & 0.94 &  & 100\% &  & 0.00 & 0.13 & 0.94 &  & - &  & - & - & - \\ 
   &  250 & -0.00 & 0.00 & 0.94 &  & 100\% &  & 0.01 & 0.05 & 0.94 & & - &  & - & - & - \\ 
   &  500 & 0.00 & 0.00 & 0.95 &  & 100\% &  & 0.00 & 0.02 & 0.94 & & - & & - & - & - \\ 
   &  750 & -0.00 & 0.00 & 0.94 &  & 100\% &  & -0.00 & 0.02 & 0.94 & & - & & - & - & - \\ 
   &  1,000 & 0.00 & 0.00 & 0.94 &  & 100\% &  & -0.00 & 0.01 & 0.94 & & - & & - & - & - \\ 
   \hline
   \multirow{1}{*}{Case 2:} & 100 & -0.00 & 0.00 & 0.93 &  & 98\% &  & 0.01 & 0.29 & 0.90 & & - & & - & - & - \\ 
   &  250 & -0.00 & 0.00 & 0.95 &  & 100\% &  & -0.02 & 0.11 & 0.93 & & - & & - & - & - \\ 
   &  500 & -0.00 & 0.00 & 0.94 &  & 100\% &  & -0.00 & 0.05 & 0.94 & & - & & - & - & - \\ 
   &  750 & -0.00 & 0.00 & 0.94 &  & 100\% &  & 0.01 & 0.03 & 0.94 & & - & & - & - & - \\ 
   &  1,000 & 0.00 & 0.00 & 0.94 &  & 100\% &  & 0.00 & 0.03 & 0.94 & & - & & - & - & - \\ 
   \hline
   \multirow{1}{*}{Case 3:} & 100 & -0.00 & 0.00 & 0.94 &  & 94\% &  & -0.02 & 0.04 & 0.90 &  & 60\% &  & 0.11 & 0.36 & 0.86 \\ 
   &  250 & 0.00 & 0.00 & 0.95 &  & 100\% &  & 0.00 & 0.01 & 0.93 &  & 88\% &  & 0.02 & 0.10 & 0.92 \\ 
   &  500 & -0.00 & 0.00 & 0.96 &  & 100\% &  & 0.00 & 0.01 & 0.94 &  & 97\% &  & -0.00 & 0.04 & 0.93 \\ 
   &  750 & -0.00 & 0.00 & 0.94 &  & 100\% &  & 0.00 & 0.00 & 0.95 &  & 99\% &  & 0.00 & 0.03 & 0.94 \\ 
   &  1,000 & -0.00 & 0.00 & 0.94 &  & 100\% &  & -0.00 & 0.00 & 0.94 &  & 100\% &  & 0.00 & 0.02 & 0.94 \\ 
   \hline
   \multirow{1}{*}{Case 4:} & 100 & -0.00 & 0.00 & 0.94 &  & 97\% &  & -0.05 & 0.05 & 0.92 &  & 64\% &  & 0.08 & 0.55 & 0.84 \\ 
   &  250 & -0.00 & 0.00 & 0.93 &  & 100\% &  & 0.01 & 0.02 & 0.93 &  & 87\% &  & 0.04 & 0.17 & 0.92 \\ 
   &  500 & 0.00 & 0.00 & 0.95 &  & 100\% &  & -0.00 & 0.01 & 0.94 &  & 96\% &  & 0.01 & 0.07 & 0.93 \\ 
   &  750 & -0.00 & 0.00 & 0.94 &  & 100\% &  & 0.00 & 0.00 & 0.94 &  & 99\% &  & 0.00 & 0.05 & 0.94 \\ 
   &  1,000 & 0.00 & 0.00 & 0.93 &  & 100\% &  & 0.00 & 0.00 & 0.94 &  & 100\% &  & -0.01 & 0.03 & 0.94 \\ 
   \hline
\end{tabular}
}
}
\end{table}

\clearpage

Table \ref{tab1} presents the results of the estimators via bootstrap and the change-point estimation procedure, while Table \ref{tab6} in Appendix \ref{appendix:study3} presents the results for the traditional MLE and the CP computed using its asymptotic distribution. The results show that the change-point detection method demonstrated consistency and improved accuracy with larger sample sizes. Models with one change point accurately identified its location, while those with two change points showed significantly better detection for the first point than for the second. Regarding bias, the estimator for the last partition consistently exhibited the least precision across all cases; however, the bootstrap method effectively corrects it. The MSE is nearly zero in all scenarios, indicating that bootstrap not only reduces the bias of the MLE but also its variance. The CP for the first partition quickly reached the specified 0.95 level, while the bootstrap confidence intervals for the last partition exhibited slightly lower coverage compared to the asymptotic distribution. 

Overall, while the change-point detection procedure works, estimating breaks in the tail becomes harder as the number of partitions grows. The bootstrap method helps correct biases in smaller samples. Finally, the asymptotic distribution of the MLE is reliable for sample sizes greater than 250.

\subsection{Study 4}\label{sec:sim_4}

In our final simulation study, we examine the performance, in terms of size and power, of the four GOF tests introduced in Section \ref{sec:gof} for the discrete piecewise power-law model. In this case, all statistics were calculated using the \texttt{dgof} package in the \texttt{R} software. We use the same cases as in Study 3, generating 5,000 pseudo-random samples of varying sizes: $n = 100, 250, 500, 750$, and 1,000.

\paragraph{Size of the test.} We estimate the empirical sizes of all tests again at the 5\% and 10\% significance levels.  Table \ref{tab2} and \ref{tab7} in Appendix \ref{appendix:study4} shows the empirical sizes for all GOF tests. As can be seen, the $W^2$, $U^2$, and $A^2$ test statistics maintain the nominal significance level and correctly discriminate between samples for small and moderate sample sizes. In contrast, the KS test rejects approximately 1\% of power-law samples across cases, once again violating the assumption that $p$-values follow a uniform distribution, resulting in an unexpected rejection rate compared to the nominal significance level.

\paragraph{Power of the test.} At this point, we assess the power of the tests under various alternative distributions. In this case, we aim to evaluate the ability of each test to detect if an incorrect number of change points is specified for datasets generated under the same cases than Study 3. Since a piecewise model with one change point is a particular case of a model with two change points, we do not compute the power in cases 1 and 2 because the tests will never detect this situation. Instead, we focus on the cases 3 and 4, estimating a model with one change point and then applying the GOF tests. Table \ref{tab2} and \ref{tab7} presents the power of the tests in parentheses at a 5\% and 10\% significance levels, respectively. 

We observed that with sample sizes smaller than 250, most tests lack the capacity to detect that a model with one change-point is incorrect when data are simulated with two change points. However, the Watson test correctly rejects the one change-point model in approximately 30\% of cases. Moreover, its power increases significantly, achieving nearly one with a sample size of 500, indicating that this test consistently rejects the piecewise model when it is incorrect.

\begin{table}[!ht]
\centering
\caption{Empirical size and power (in parentheses) of the $K$, $W^2$, $U^2$, and $A^2$ statistics at a 5\% significance level for discrete piecewise power-law samples across different cases.}
\label{tab2}
\scalebox{0.83}{
\begin{tabular}{cccccccccccc}
  \hline
  & \multicolumn{4}{c}{Case 1} & & \multicolumn{4}{c}{Case 2} \\
  \cline{2-5} \cline{7-10}
  $n$ & $K$ & $W^2$ & $U^2$ & $A^2$ & & $K$ & $W^2$ & $U^2$ & $A^2$ \\
  \hline
   100 & 0.01 & 0.04 & 0.05 & 0.04 &  & 0.01 & 0.06 & 0.05 & 0.05 \\ 
   250 & 0.01 & 0.04 & 0.03 & 0.04 &  & 0.01 & 0.04 & 0.04 & 0.04 \\ 
   500 & 0.02 & 0.04 & 0.05 & 0.04 &  & 0.01 & 0.04 & 0.05 & 0.05 \\ 
   750 & 0.01 & 0.04 & 0.04 & 0.04 &  & 0.01 & 0.04 & 0.04 & 0.04 \\ 
   1,000 & 0.02 & 0.04 & 0.04 & 0.05 &  & 0.02 & 0.05 & 0.05 & 0.05 \\  
   \hline
   & \multicolumn{4}{c}{Case 3} & & \multicolumn{4}{c}{Case 4} \\
   \cline{2-5} \cline{7-10}
   $n$ & $K$ & $W^2$ & $U^2$ & $A^2$ & & $K$ & $W^2$ & $U^2$ & $A^2$ \\
   \hline
   100 & 0.01 (0.02) & 0.04 (0.02) & 0.05 (0.32) & 0.04 (0.01) &  & 0.02 (0.01) & 0.04 (0.01) & 0.05 (0.26) & 0.04 (0.01) \\ 
  250 & 0.01 (0.49) & 0.04 (0.39) & 0.04 (0.90) & 0.04 (0.45) &  & 0.01 (0.26) & 0.04 (0.23) & 0.04 (0.82) & 0.04 (0.29) \\ 
  500 & 0.02 (0.98) & 0.04 (0.96) & 0.04 (1.00) & 0.04 (0.98) &  & 0.02 (0.87) & 0.04 (0.88) & 0.04 (0.99) & 0.04 (0.94) \\ 
  750 & 0.02 (1.00) & 0.04 (1.00) & 0.04 (1.00) & 0.04 (1.00) &  & 0.01 (0.99) & 0.04 (0.99) & 0.05 (1.00) & 0.04 (1.00) \\ 
  1,000 & 0.02 (1.00) & 0.05 (1.00) & 0.04 (1.00) & 0.05 (1.00) &  & 0.02 (1.00) & 0.05 (1.00) & 0.04 (1.00) & 0.05 (1.00) \\ 
   \hline
\end{tabular}}
\end{table}

\paragraph{Conclusion.} The Watson test is the most powerful test among those that maintain the nominal significance level, effectively identifying situations where fewer change points are estimated than needed. This makes it highly recommended for practical situations.

\section{Applications}\label{sec:application}

In this section, we analyze two real datasets to demonstrate the potential of the proposed methodology for the power-law model and highlight the advantages of its piecewise extension.

\subsection{Index of complex networks}

In a recent study, \cite{broido2019scale} analyzed nearly 1,000 real-world networks from biological, informational, social, technological, and transportation domains obtained from the Index of Complex Networks (ICON) \citep{clauset2016colorado} to verify the scale-free hypothesis in each case. This corpus of networks includes both simple graphs and networks with various combinations of directed, weighted, bipartite, multigraph, temporal, and multiplex properties. They developed a set of graph simplification functions to convert a network into a set of simple graphs, and extract their degree sequences (see Supplementary Information in \citeauthor{broido2019scale}, \citeyear{broido2019scale} for details).

They applied the methodologies from \cite{clauset2009power} to analyze the degree sequences obtained. For a given degree distribution, they employed a KS statistic minimization to determine the minimum degree $x_{\mathrm{min}}$ above which the degrees are most closely modeled by a power-law distribution, thereby truncating non-power-law behavior among low-degree nodes. The scaling parameter was estimated using the ML estimator, and the statistical plausibility of the fitted model was assessed using the standard semiparametric bootstrap approach based on KS statistics. Finally, networks were categorized according to a taxonomy of different notions of evidence for scale-free structure outlined in Table \ref{tab3}.

\clearpage

\begin{table}[h]
    \centering
    \caption{Taxonomy of scale-free network definitions in \cite{broido2019scale}.}
    \label{tab3}
    \scalebox{0.93}{
    \begin{tabular}{rl}
        \hline
        Classification & Description \\
        \hline
        Super-Weak & For at least 50\% of graphs, no alternative distribution is favored over the power-law. \\
        Weakest & For at least 50\% of graphs, a power-law distribution cannot be rejected ($p \geq 0{.}1$). \\
        Weak & Requirements of Weakest, and the power-law region contains at least 50 nodes. \\
        Strong & Requirements of Weak and Super-Weak, and $2 < \alpha < 3$ for at least 50\% of graphs. \\
        Strongest & Requirements for Strong in at least 90\% of the graphs and for Super-Weak in at least 95\%\\
        Not Scale Free & Networks that are neither Super-Weak nor Weakest. \\
        \hline
    \end{tabular}
    }
\end{table}

The analysis reveals that only 4\% of networks strongly support scale-free structure and 29\% meet weaker criteria, while 46\% of networks indicate that the power law is simply a statistically better fit than alternatives, and 49\% show no evidence of scale-free structure. These findings challenge previous assumptions that power-law degree distributions are universal across various scientific domains \citep{barabasi2016network}. However, although this study provides a broad analysis, its results have certain drawbacks. Different parameter values result in varying classifications within the scale-free taxonomy, and the estimator itself is biased (as we discussed in Section \ref{sec:sim_1}). Moreover, the GOF test used in this analysis tends to reject the power-law model even when the samples are generated from the same model (see Section \ref{sec:sim_2}).

In order to address this issue, we apply our proposed methodology to the same degree sequences, using the MAP estimator with the Jeffreys prior for the discrete power-law model to estimate the scaling parameter, and then applying the Watson test to determine whether the degree sequences follow a power-law at a 10\% significance level, following the weakest taxonomy. The results are shown in Table \ref{tab4}.

\begin{table}[ht]
\centering
\caption{Number of degree sequences in each of the five domains provided by ICON for which the power law is valid, according to the weakest taxonomy under the four GOF statistical tests, at a 10\% significance level.}
\label{tab4}
\begin{tabular}{lccccc}
    \hline
    Domain & Number &$K$ & $W^2$ & $U^2$ & $A^2$ \\ 
    \hline
    Biological & 893 & 258 & 362 & 324 & 354 \\ 
    Informational &  26 &   8 &  12 &   9 &  12 \\ 
    Social & 626 & 289 & 437 & 343 & 429 \\ 
    Technological & 1,953 & 1,486 & 1,746 & 1,678 & 1,709 \\ 
    Transportation & 171 &   4 &   8 &   5 &  10 \\
    \hline
    Total & 3,669 & 2,045 & 2,565 & 2,359 & 2,514 \\ 
    \hline
\end{tabular}
\end{table}

We observe that across all domains, the power law is validated in more sequences using the alternative GOF tests than the KS test. Specifically, under the Watson test, more than 300 degree sequences adhere to the power-law distribution. Additionally, we plot the scaling parameters estimated by the MAP estimator in Figure \ref{fig5}, where we can observe that most of the degree sequences align with the scale-free hypothesis. However, there is a clear limitation: we cannot identify all the sequences corresponding to a specific network in every case from their GitHub repository. Consequently, we are limiting our analysis to the ungrouped sequences. Nonetheless, it is clear that these proposed methodologies improve the assessment across the networks considered, as they provide a more accurate and robust way to evaluate the presence of scale-free structures, minimizing false negatives seen in traditional KS tests.

\begin{figure}[H]
    \centering
    \includegraphics[scale=0.6]{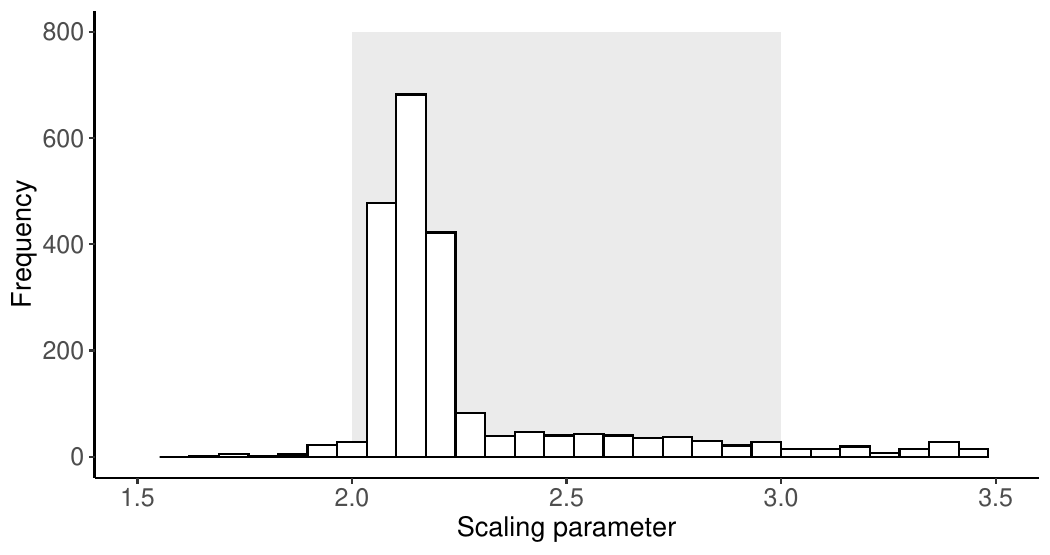}
    \caption{Histogram of the estimated scaling parameter, derived from equation \eqref{eqn:map1}, for power-law degree sequences under the Watson test, highlighting the range where the strong scale-free evidence taxonomy holds.}
    \label{fig5}
\end{figure}

In Figure \ref{fig6}, we plot the degree sequences based on their average degree (the mean number of connections per node in the sequence) and their size (the total number of nodes), emphasizing the points where the two tests under study do not yield the same results in determining whether a sequence adheres to the power law, i.e., with one test rejecting the hypothesis while the other accepts it.

\begin{figure}[H]
    \centering
    \includegraphics[scale=0.45]{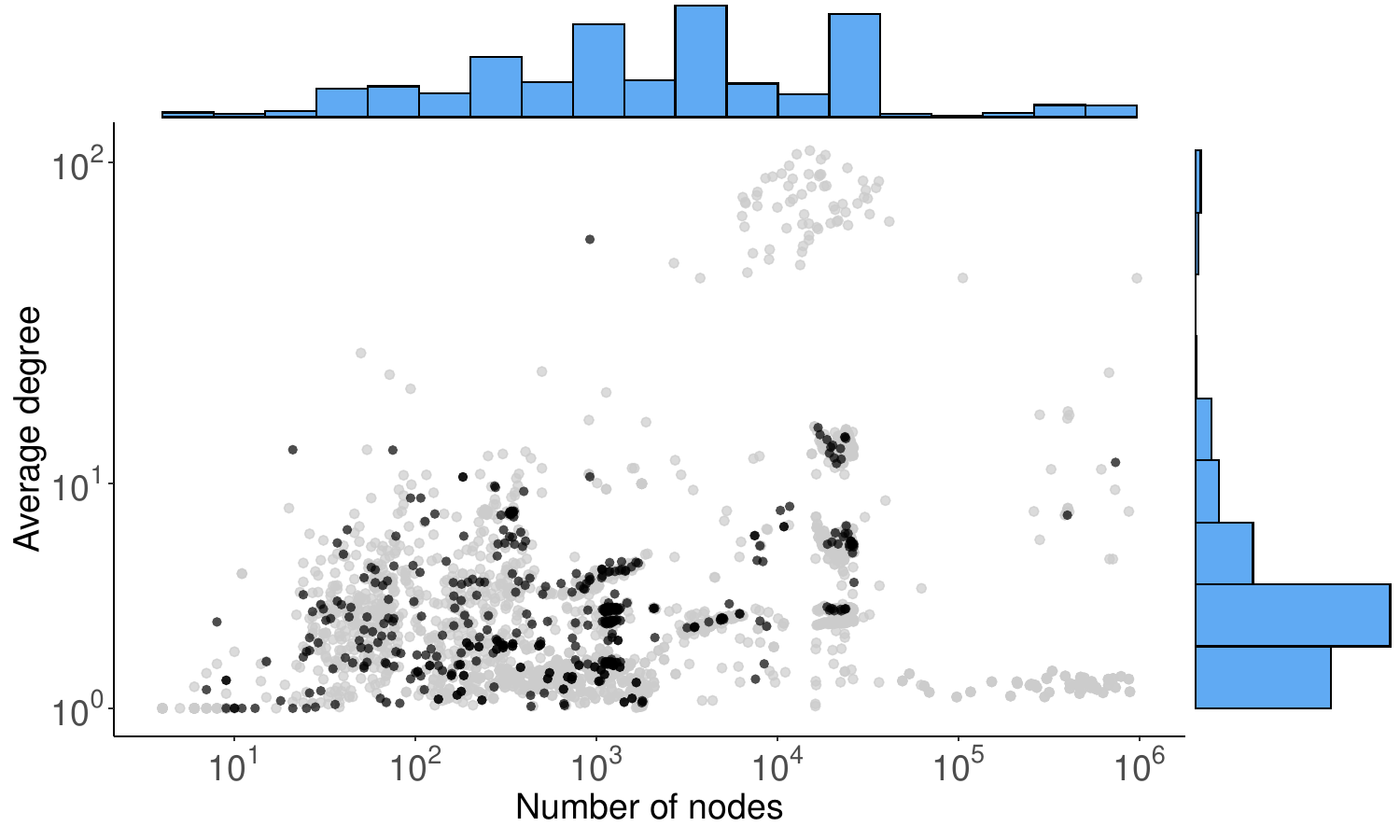}
    \caption{Average degree of a degree sequence as a function of the number of nodes, highlighting with dark points those sequences where the KS and Watson tests yield different results.}
    \label{fig6}
\end{figure}

As observed, most of the sequences yield consistent results across both tests. However, there are 494 sequences that exhibit different classifications. Specifically, 404 of them are considered to follow a power-law distribution by the Watson test but are not classified as such by the KS test. This situation coincides with our findings in Section \ref{sec:sim_2}, where we observed a greater tendency to reject power-law behavior under the KS test in datasets with small and moderate sample sizes.

In order to propose an alternative approach to analyze these networks, we will assess the performance of the piecewise power-law model by applying it to two degree sequences, selected randomly, from biological networks that have rejected the scale-free hypothesis providing insights into the applicability of semiparametric models in complex networks.

\paragraph*{Binary Interactomes}

The binary interaction network is part of a collection of networks that describe protein-protein interactions across 11 species  \citep{das2012hint}. The networks are undirected and unweighted, and were created by combining data from high-throughput experiments, with manually curated information from literature. We focus on a degree sequence for the species \textit{Mus musculus}. 

Using the proposed methodology to fit power-law models, we found a lower bound equal to 2, which discards around 36\% of the information, and an estimated scaling parameter of 2.51. However, the scale-free hypothesis is rejected in this sequence based on the KS test and the Watson test, with $p$-values given by $p(K) < 0.01$ and $p(U^2) = 0.07$. On the other hand, we fit the piecewise power-law model by introducing a change point until the model is no longer rejected by the Watson test. In this case, the model estimates a change point located at 2, segmenting the power-law behavior with a scaling parameter of 1.46 before this point and 2.51 after it (see panel (a) of Figure \ref{fig7} for details). Curiously, including the entire information provided in the degree sequence allows us that the scale-free hypothesis is in fact accepted, with $p(U^2)= 0.29$, because the piecewise model is an extension of the power law, which it begins to be plausible from the last change point, we have shown in Proposition \ref{prop3}. 

\paragraph*{Rat Brain} This network captures interactions between neurons of a rat brain traced through large-scale electron microscopy of serial thin sections, with nodes representing anatomical regions and edges indicating neural interactions \citep{bota2007online}. The data were collected by the Brain Architecture Management System project. We focus on the degree sequence derived from this network for the species \textit{Rattus norvegicus}. 

If we fit a power-law model, we obtain a lower bound of 6, discarding around 64\% of the information, and the estimated scaling parameter is 2.31. The scale-free hypothesis is again rejected by both tests ($p(K) < 0.01$, $p(U^2) = 0.08$). However, if we fit a piecewise model, we found that the model identifies two change points located at 2 and 6, segmenting the power-law behavior into scales of 1.31, 1.55, and 2.32 (see panel (b) of Figure \ref{fig7}). One more time, the inclusion of the entire information provided in the degree sequence allows us that the scale-free hypothesis is in fact accepted, with $p(U^2)= 0.61$.

This analysis demonstrates the advantages of the piecewise power-law model in capturing complex behaviors in degree sequences that initially reject the scale-free hypothesis when fitted with the standard power-law model. By introducing change points, the piecewise model provides a more flexible framework that better accommodates the underlying structure of the data. Moreover, in both cases, the selected piecewise model includes a change point immediately following the minimum value equal to 1, demonstrating its ability to explain values-inflated datasets. Such situations can be modeled by values-inflated models, which are statistical tools designed to account for an excessive frequency of certain values in a dataset, typically zeros or ones \citep{cheung2002zero}. These models provide a more accurate representation of the data by incorporating the excess mass at these specific values.

\begin{figure}[H]
    \centering
    \begin{subfigure}{0.45\textwidth}
        \centering
        \includegraphics[width=\linewidth]{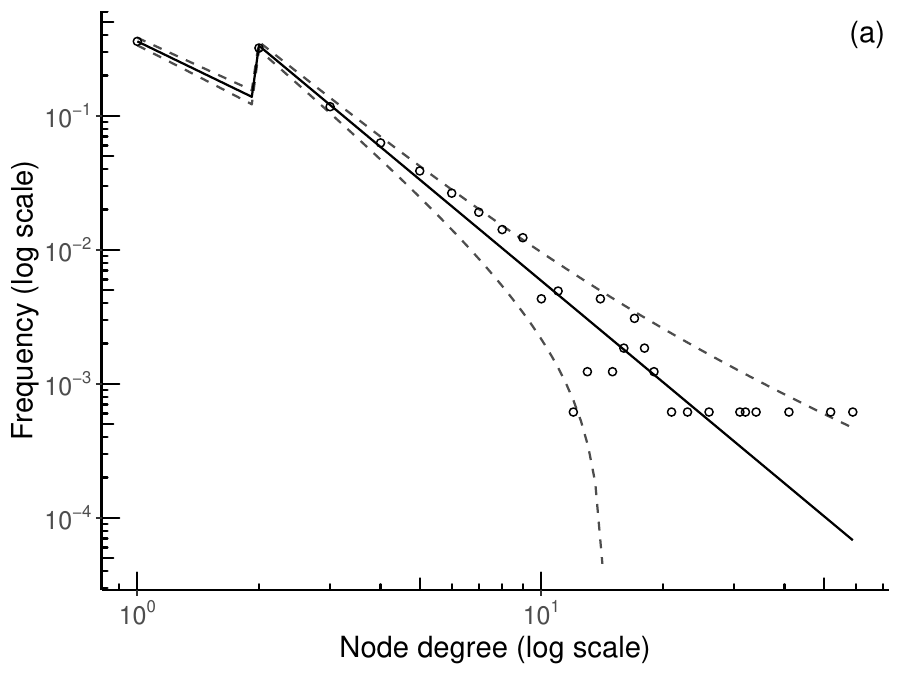}
    \end{subfigure}
    \hspace{1cm}
    \begin{subfigure}{0.45\textwidth}
        \centering
        \includegraphics[width=\linewidth]{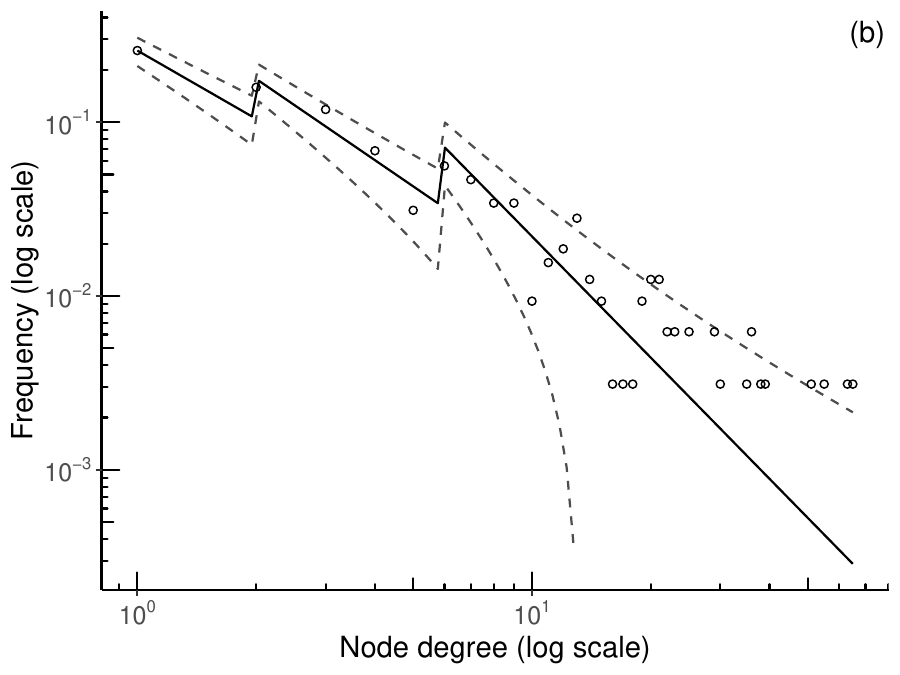}
    \end{subfigure}
    \caption{Degree distributions of two biological networks: (a) interactions between proteins; and (b) interactions between anatomical regions. The estimated probability mass function is shown as a solid line, with its 95\% confidence intervals (based on the normal distribution approximation for large samples) shown as a dashed line.}
    \label{fig7}
\end{figure}

\subsection{Word frequency in Moby Dick}

%Although the primary motivation for this work arises from studying the scale-free hypothesis in the degree distribution of complex networks, our approach is not limited to this application, and we address these issues within the broader context of power-law models.

The second dataset considered in this work consists of the word frequency distribution in the English text of Herman Melville's novel \textit{Moby Dick}, in order to demonstrate the potential of the semiparametric model beyond the domain of complex networks. This book, published in 1851, is a masterpiece of American literature, renowned for its richly descriptive language and thematic complexity. The novel's diverse vocabulary, ranging from nautical terms to profound philosophical reflections, offers an excellent dataset to study linguistic patterns. Figure \ref{fig8} presents the 150 most frequent words in this text, after a traditional text cleaning process that included removing punctuation and stopwords.

\begin{figure}[H]
    \centering
    \includegraphics[scale=0.6]{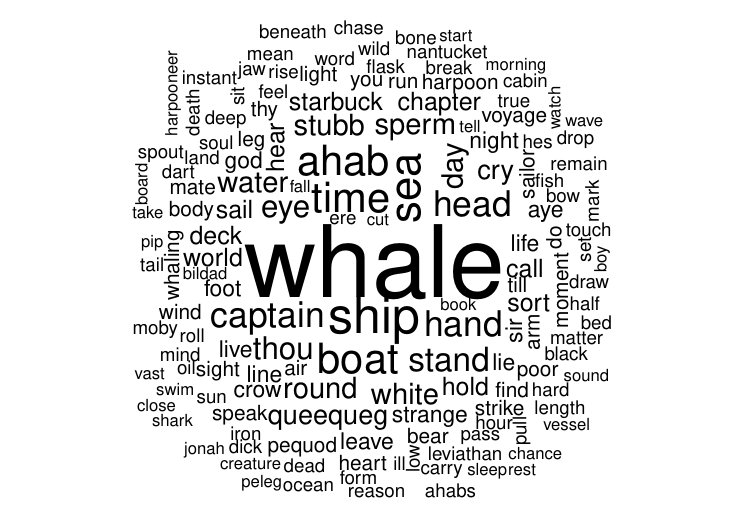}
    \caption{Word cloud of the 150 most frequent words in the novel Moby Dick.}
    \label{fig8}
\end{figure}

\cite{jean1916gammes} observed that the frequency of word usage in the novel appears to follow a power-law distribution, indicating that a few words dominate communication while many other words are used less frequently but contribute to making the language more expressive and adaptable. 

\cite{newman2005power} analyzed this dataset using the continuous version of the power-law model, fixing the lower bound to the minimum value of the data ($x_{\min} = 1$) and applying a Bayesian approach for the scaling parameter with a uniform prior, obtaining $\alpha = 2.20$ as the MAP estimate. In contrast, by applying the methodologies of \cite{clauset2009power}, which respect the discrete nature of the context, we find that the lower bound where the power-law becomes valid is $x_{\min} = 7$, with a scaling parameter of $\alpha = 1.95$. The $p$-values corresponding to the GOF tests evaluated in this work are $p(K) = 0.48$, $p(W^2) = 0.74$, $p(U^2) = 0.58$, and $p(A^2) = 0.73$, which do not provide evidence against this model.

However, when we apply the power-law model starting from $x_{\min}=7$, we are discarding 84\% of the information above this value, which is problematic because it leads an incomplete understanding of the data's overall structure. In order to model the entire dataset, we will use the piecewise power-law model. Since the exact number of change points is uncertain, we fit three piecewise models with one, two, and three change points, respectively, and evaluate their performance using the AIC and BIC. The results, including the parameter estimates and the locations of the change points, are presented in Table \ref{tab5}. According to the AIC, the best-fitting model is the piecewise model with three change points. However, to balance model fit and simplicity, we prioritize the BIC results, opting for a model with two change points to avoid unnecessary complexity and prevent overfitting.

\begin{table}[!ht]
\centering
\caption{Performance of the discrete piecewise power-law (PWPL) model  with one, two, and three change points. The table presents the estimated parameters for each segment, the change points, and the AIC and BIC values.}
\label{tab5}
\scalebox{1}{
\begin{tabular}{l|cccc|cccc|cc}
    \hline
    Model & $\alpha_1$ & $\alpha_2$ & $\alpha_3$ & $\alpha_4$ & $\tau_{(0)}$ & $\tau_{(1)}$ & $\tau_{(2)}$ & $\tau_{(3)}$ & AIC & BIC \\ 
    \hline
    PWPL$(k = 1)$ & $1.72$ & $1.89$ & - & - & 1 & 3 & - & - & 80,206.43 & 80,229.97 \\ 
    PWPL$(k = 2)$ & $1.70$ & $1.81$ & $1.94$ & - & 1 & 2 & 6 & - & 80,175.64 & 80,214.86 \\ 
    PWPL$(k = 3)$ & $1.70$ & $1.78$ & $1.85$ & $1.95$ & 1 & 2 & 3 & 7 & 80,168.09 & 80,223.00 \\
   \hline
\end{tabular}
}
\end{table}

Figure \ref{fig9} shows the rank-frequency plot, the fitted piecewise power-law model.

\begin{figure}[H]
    \centering
    \includegraphics[scale=0.6]{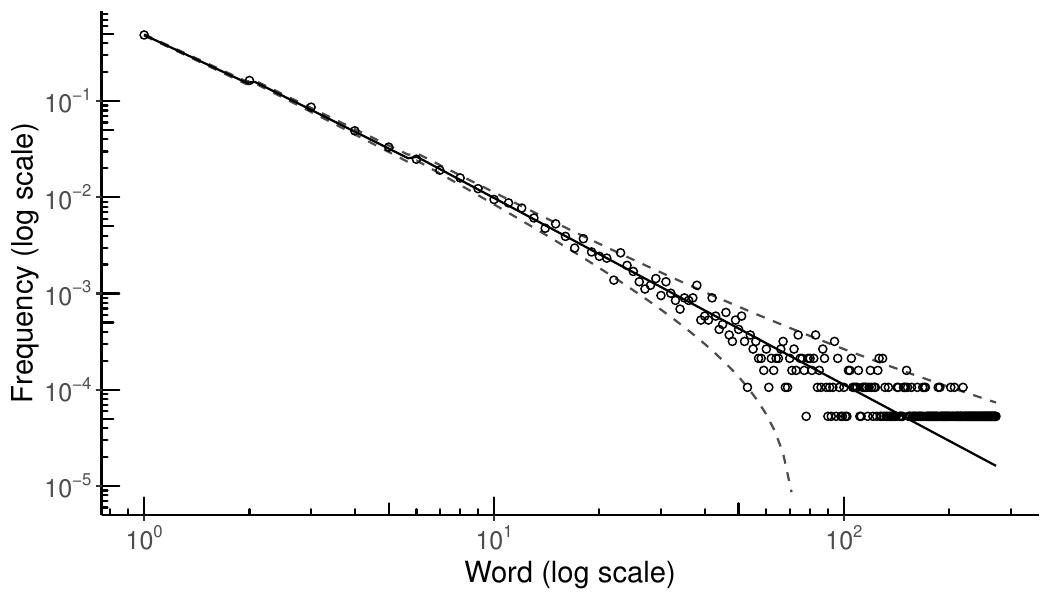}
    \caption{Rank/frequency plot of the numbers of occurrences of words in the novel Moby Dick, along with the probability mass function (solid line) and its 95\% confidence intervals (dashed line), of the fitted model.}\label{fig9}
\end{figure}

The results indicate that the proposed model nearly overlaps the points, generalizing the findings of the model without breaks, which basically corresponds to the last segment. In this case, the $p$-values are $p(K) = 0.49$, $p(W^2) = 0.48$, $p(U^2) = 0.32$, and $p(A^2) = 0.52$, none of which provide evidence against this new model. From this figure, it can be seen that the relationship between words and their frequency is not exactly linear (in the log-log scale). Instead of a single peak of high-frequency words, we might observe different zones where the power-law behavior fits with different parameters.

\section{Conclusion}\label{sec:conclusion}

In this paper, we have conducted a comprehensive analysis of two critical aspects of statistical modeling: parameter estimation and GOF evaluation on the discrete power-law model and its piecewise extension, which provides a more accurate representation of the underlying data distribution by accommodating potential changes in the scaling behavior across different segments. 

Firstly, we compare the performance of different estimators using both classical and Bayesian methods, assuming an improper Jeffreys prior obtained with the continuous and discrete versions of the power-law model. Through a simulation study, we found that frequentist methods for estimating the scaling parameter often produce biased estimates, particularly when the sample size is small. In contrast, the Bayesian estimator provides unbiased estimates, demonstrating superior accuracy in such scenarios. However, this improvement in estimation precision comes at the cost of increased variance. Despite this trade-off, the credible intervals generated by the Bayesian estimator are more reliable, as they tend to align closely with the specified nominal confidence level. This makes the Bayesian approach a more robust alternative compared to frequentist methods.

In principle, estimating the parameters of the model is not only part of the inferential process but also essential for validating the fitted distribution, i.e., determining whether empirical data can be described by a discrete power-law distribution. In the literature, this task is typically performed using the KS test with a semiparametric bootstrap to estimate its distribution and compute the $p$-value. Our findings suggest that this method does not provide reliable results. In general, the test tends to reject samples, even when they are generated from the same power-law model, leading to a high probability of committing a type I error. As a result, the test appears more powerful than it actually is. To address this issue, we evaluated three well-known GOF tests from the literature. Our results show that the Watson test stands out due to its superior power while successfully controlling the type I error rate. This finding is also reported in the literature with other distributions (see \citeauthor{pewsey2018parametric}, \citeyear{pewsey2018parametric}).

Through the analysis of degree sequences of nearly 1,000 complex networks from different domains, we applied our proposed methodology. Specifically, we used the Bayesian estimator to estimate the scaling parameter in each step of the search of the lower bound where the power law applies, and the Watson test to decide whether a sequence is well-explained by this model. Our main finding is that this new methodology does not reject the scale-free hypothesis in many sequences where the KS test does. The main discrepancies were found in small and moderate samples, as the simulation study indicates. Overall, our methodology can contribute to a more precise evaluation of the scale-free hypothesis, which is important because it helps in understanding the underlying structure of networks.

One limitation of the power-law model is that it typically starts from a minimum value, which does not necessarily coincide with the minimum of the sample. This results in a significant portion of the data being excluded from the estimation procedure, and as a result, the underlying mechanism of the data is not adequately captured. One way to address this issue is to assume that different power-law patterns exist within the data, instead of a single power-law. This novel piecewise model enhances flexibility and enables the modeling of the entire dataset, rather than just the tail; indeed, it can effectively recover dataset with values-inflated, making it highly useful and interpretable.

We assess the estimation of this model, finding that the MLEs produce unbiased estimates except for the last interval, which is expected since the last segment of the piecewise model recovers the standard power-law model. Given that the Fisher information has a complex formulation in the piecewise model, producing Bayesian estimators in this model and proving that the posterior is proper can be a very complex task. Instead, we proceed with a bootstrap scheme to estimate the bias and correct the estimators. The results show that the bias-corrected bootstrap MLE leads to unbiased estimations. In addition, the confidence intervals produced by the asymptotic distribution of the MLE and the bootstrap confidence intervals provide precise coverage probabilities for sample sizes of 250 or more, at the same time that we are estimating the change point through an exhaustive search.

We demonstrate the benefits of our approach by fitting the piecewise model to analyze three datasets: two degree sequences from biological complex networks where the scale-free hypothesis is rejected, and the word frequency distribution from the novel \textit{Moby Dick}, which is known to follow a power-law distribution. We found that all datasets can be effectively modeled in a flexible way with this model. Moreover, the piecewise model can accommodate the scale-free hypothesis in networks where the power-law model is rejected, opening a new avenue for discussing the limitations of assessing the scale-free hypothesis with the traditional power-law model. This motivates the development of semiparametric approaches for analyzing complex networks, offering a complementary perspective to \cite{broido2019scale} by applying our new methodologies to better understand the prevalence of scale-free networks in the real world. 

Additionally, these real-world examples highlight the critical need for using rigorous criteria to estimate the number of change points of a piecewise model. While the AIC often favors models with a greater number of change points due to its focus on improving model fit, this approach can lead to overfitting, where the model captures spurious variations rather than genuine structural changes. In contrast, the BIC introduces a more conservative penalty for complexity, promoting models that balance fit and interpretability. This balance highlights the importance of choosing criteria that keep the model simple, which is key to understanding important patterns in the data.

Future research could focus on expanding upon the findings of this study in several directions. Instead of modeling previous values of the lower bound with several power-law scales, it may be beneficial to explore nonparametric approaches, such as using splines or similar techniques. This enhancement could result in a more versatile model capable of explaining a wider range of datasets. Another important concern is the inclusion of the covariates, which could better capture the transitions between different power-law patterns, enabling more precise modeling and facilitating comparisons across individuals based on their factors or attributes. This is motivated because it allows us to understand how other factors influence the relationship between the response variable, providing a more adaptable to different contexts.

%\section*{Acknowledgements}

% The authors thank the associate editor and two anonymous referees for their constructive comments and suggestions, which led to improvements in a previous version of the manuscript

\clearpage
\section*{Disclosure statement}

No potential conflict of interest was reported by the authors.

\section*{Funding}

Nixon Jerez-Lillo was funded by the National Agency for Research and Development (ANID)/Scholarship Program/Doctorado Nacional/2021-21210981. Francisco A. Rodrigues is indebted to Brazilian National Council for Scientific and Technological Development (CNPq, Grant 308162/2023-4) and The State of São Paulo Research Foundation (FAPESP, Grants 20/09835-1 and 13/07375-0) for the financial support provided to this research. Paulo H. Ferreira acknowledges support from CNPq (Grant 307221/2022-9). 

\section*{Software}

The research was carried out using the computational resources of the Center for Mathematical Sciences Applied to Industry (CeMEAI) funded by FAPESP (Grant 2013/07375-0).

All the functions and procedures concerning implementation included in the article have been implemented in \cite{Rsoftware}. The code for the discrete piecewise power-law model was inspired by the implementation of the \texttt{PWEXP} package in \texttt{R} \citep{xu2024pwexp}. The codes will be available at \texttt{https://github.com/...}

\bibliographystyle{plainnat}
\bibliography{references}

\appendix

\section{Proof of Proposition \ref{prop1}}\label{appendix:proof1}

From complex analysis, we know that $\zeta(\alpha, x_{\min})$ is a meromorphic function with a simple pole at $\alpha = 1$ and no other poles for $\alpha > 1$. As shown in \cite{stein2010complex} (p. 74, Theorem 1.2), there exists an $\epsilon > 0$ and a non-vanishing analytic function $g(\alpha)$ on the interval $(1-\epsilon, 1+\epsilon)$ such that:
\begin{align*}
    \zeta(\alpha,x_{\min}) = (\alpha-1)^{-1} g(\alpha).
\end{align*}

%This behavior indicates the function's singularity structure.

Moreover, since $\alpha$ is a real number, the function $\zeta(\alpha, x_{\min})$ also takes real values, and therefore, $g(\alpha)$ is real for all $\alpha$ in the interval $(1-\epsilon, 1+\epsilon)$. Thus, it follows that:
\begin{equation*}
\lim_{\alpha\rightarrow 1^+} \frac{\zeta^{-1}(\alpha,x_{\min})}{\alpha^{-1}(\alpha-1)x_{\min}^{\alpha}}=\lim_{\alpha\rightarrow 1^+} \dfrac{[g(\alpha)]^{-1}}{x_{\min}^{\alpha} \alpha^{-1}}=x_{\min}^{-1}\,\left[g(1)\right]^{-1}<+\infty.
\end{equation*}

Besides, we have:
\begin{align*}
\log \zeta(\alpha,x_{\min})=-\log(\alpha-1)+\log g(\alpha),
\end{align*}
for all $\alpha \neq 1$ in the interval $(1-\epsilon, 1+\epsilon)$. We can differentiate twice to obtain:
\begin{gather*}
\phi^{(2)}(\alpha,x_{\min})=\frac{1}{(\alpha-1)^{2}}
+\frac{\partial^2 \log g(\alpha)}{ \partial \alpha^2},\\ 
\Rightarrow \lim_{\alpha\rightarrow 1^+}\frac{\phi^{(2)}(\alpha,x_{\min})}{\alpha^2(\alpha-1)^{-2}}=\lim_{\alpha\rightarrow 1^+}\left[\frac{1}{\alpha^2} + \frac{(\alpha-1)^{2}}{\alpha^2}\frac{\partial^2 \log g(\alpha)}{ \partial \alpha^2}\right]=1,
\end{gather*}
where the last equality follows directly from the fact that the second derivative of $\log g(\alpha)$ with respect to $\alpha$ is continuous in the interval $(1-\epsilon, 1+\epsilon)$, thereby permitting the algebra of limits. On the other hand, note that as $\alpha$ becomes large, the first term $x_{\min}^{-\alpha}$ dominates the series because the subsequent terms $(k + x_{\min})^{-\alpha}$ for $k \geq 1$ decay rapidly and become negligible. Therefore:
\begin{gather*}
\lim_{\alpha\rightarrow+\infty}\frac{\zeta(\alpha,x_{\min})}{x_{\min}^{-\alpha}}=1,
\end{gather*}
and then:
\begin{equation*}
\lim_{\alpha\rightarrow+\infty}\frac{\zeta^{-1}(\alpha,x_{\min})}{\alpha^{-1}(\alpha-1)x_{\min}^{\alpha}}=1.
\end{equation*}

Analogously, taking into account:
\begin{align*}
-\zeta'(\alpha, x_{\min}) &= \sum_{k=0}^{+\infty}\log(k+x_{\min})\left(k+ x_{\min}\right)^{-\alpha},\\ 
\zeta''(\alpha, x_{\min}) &=   \sum_{k=0}^{+\infty}\left[\log(k+x_{\min})\right]^2\left(k+ x_{\min}\right)^{-\alpha},
\end{align*}
we obtain:
\begin{align*}
\lim_{\alpha\rightarrow+\infty}\frac{\zeta'(\alpha,x_{\min})}{x_{\min}^{-\alpha}}=-\log(x_{\min}) \quad \mbox{ and } \quad \lim_{\alpha\rightarrow+\infty}\frac{\zeta''(\alpha,x_{\min})}{x_{\min}^{-\alpha}}=\left[\log(x_{\min})\right]^2.
\end{align*}

Hence, we conclude that:
\begin{align*}
\lim_{\alpha\rightarrow+\infty} \frac{\phi^{(2)}(\alpha,x_{\min})}{\alpha^2(\alpha-1)^{-2}}
&=\lim_{\alpha\rightarrow+\infty}\left[\left(\frac{\zeta''(\alpha,x_{\min})}{x_{\min}^{-\alpha}}\right)\left(\frac{\zeta(\alpha,x_{\min})}{x_{\min}^{-\alpha}}\right)^{-1}\dfrac{(\alpha-1)^2}{\alpha^2}\right.-\\
&\hspace{2cm}\left.\left(\frac{\zeta'(\alpha,x_{\min})}{x_{\min}^{-\alpha}}\right)^2\left(\frac{\zeta(\alpha,x_{\min})}{x_{\min}^{-\alpha}}\right)^{-2}\dfrac{(\alpha-1)^2}{\alpha^2}\right], \\
&=\left[\log(x_{\min})\right]^2-\left[\log(x_{\min})\right]^2, \\
&=0.
\end{align*}

Using Proposition \ref{prop0}, there exist constants $C_1 > 0$ and $C_2 > 0$ such that:
\begin{align*}
\zeta^{-1}(\alpha,x_{\min})\leq C_1\cdot \alpha^{-1}(\alpha-1)x_{\min}^{\alpha} \quad \mbox{ and } \quad \phi^{(2)}(\alpha,x_{\min})\leq C_2\cdot \alpha^2(\alpha-1)^{-2}.
\end{align*}

Based on these results, it follows that:
\begin{align*}
\int_{1}^{+\infty}\alpha^r\frac{\sqrt{\phi^{(2)}(\alpha,x_{\min})}}{\left[\zeta(\alpha,x_{\min})\right]^{n}}\prod_{i=1}^{n}x_{i}^{-\alpha} d\alpha &\leq C_1^n\sqrt{C_2}\int_{1}^{+\infty}\alpha^{r-(n-1)}(\alpha-1)^{n-1}\prod_{i=1}^{n}\left(\frac{x_{i}}{x_{\min}}\right)^{-\alpha}d\alpha,\\
&\leq C_1^n\sqrt{C_2}\int_{1}^{+\infty}\alpha^{r}\prod_{i=1}^{n}\left( \frac{x_i}{x_{\min}}\right)^{-\alpha} \, d\alpha,\\
&\leq C_1^n\sqrt{C_2} \int_{0}^{+\infty}\alpha^{r}\exp\left\{-\alpha \sum_{i = 1}^{n} \log\left(\frac{x_i}{x_{\min}}\right)\right\} \, d\alpha,\\
&= C_1^n\sqrt{C_2}\,\Gamma(r+1) \left(\sum_{i = 1}^{n} \log\left(\frac{x_i}{x_{\min}}\right)\right)^{-(r+1)},
\end{align*}
which proves that the posterior distribution is finite.

\section{Proof of Proposition \ref{prop2}}\label{appendix:proof2}

Using the identities proved in Proposition \ref{prop1} and following similar steps, we observe that:
\begin{align*}
\int_{1}^{+\infty}\alpha^r\frac{(\alpha-1)^{-1}}{\left[\zeta(\alpha,x_{\min})\right]^{n}}\prod_{i=1}^{n}x_{i}^{-\alpha} d\alpha
&\leq C_1^n\int_{1}^{+\infty}\alpha^{r-n}(\alpha-1)^{n-1}\prod_{i=1}^{n}\left(\frac{x_{i}}{x_{\min}}\right)^{-\alpha}d\alpha,\\
&\leq C_1^n\int_{1}^{+\infty}\alpha^{r-1}\exp\left\{-\alpha \sum_{i = 1}^{n} \log\left(\frac{x_i}{x_{\min}}\right)\right\}d\alpha,\\
&\leq C_1^n\int_{0}^{+\infty}\alpha^{r-1}\exp\left\{-\alpha \sum_{i = 1}^{n} \log\left(\frac{x_i}{x_{\min}}\right)\right\}d\alpha,\\
&=C_1^n\,\Gamma(r) \left(\sum_{i = 1}^{n} \log\left(\frac{x_i}{x_{\min}}\right)\right)^{-r},
\end{align*}
which proves that this new posterior distribution is also finite.

\section{Calculation of the derivatives of the Hurwitz zeta function}\label{appendix:derivatives_HZfunction}

Taking the first derivative with respect to $\alpha$, we get:
\begin{align*}
\zeta'(\alpha,x_{\min}) = \frac{\partial}{\partial \alpha} \left( \sum_{k=0}^{+\infty} (k + x_{\min})^{-\alpha} \right) = \frac{\partial}{\partial \alpha} \left( \sum_{k=x_{\min}}^{+\infty} k^{-\alpha} \right) = -\sum_{k=x_{\min}}^{+\infty} \frac{\log(k)}{k^\alpha}.
\end{align*}

Now, note that:
\begin{align*}
\zeta'(\alpha,x_{\min}) = -\sum_{k=1}^{+\infty} \frac{\log(k)}{k^\alpha} + \sum_{k=1}^{x_{\min}-1} \frac{\log(k)}{k^\alpha} = \zeta'(\alpha) + \sum_{k=1}^{x_{\min}-1} \frac{\log(k)}{k^\alpha}.
\end{align*}

For the second derivative, we differentiate $\zeta'(\alpha,x_{\min})$ again with respect to $\alpha$:
\begin{align*}
\zeta''(\alpha,x_{\min}) &= \frac{\partial}{\partial \alpha} \left( \zeta'(\alpha) + \sum_{k=1}^{x_{\min}-1} \frac{\log(k)}{k^\alpha} \right),\\
&= \zeta''(\alpha) + \sum_{k=1}^{x_{\min}-1} \frac{\partial}{\partial \alpha} \left( \frac{\log(k)}{k^\alpha} \right),\\
&= \zeta''(\alpha) - \sum_{k=1}^{x_{\min}-1} \frac{\left[\log(k)\right]^2}{k^\alpha}.
\end{align*}

This allows us to test the calculation of both derivatives

\section{Proof of Proposition \ref{prop3}}\label{appendix:proof3}

First of all, we must note that:
\begin{equation*}
\mathbb{P}(X \geq x)=\sum_{j=1}^{k + 1} \left[ \frac{\zeta(\alpha_j,x)}{\zeta(\alpha_{j},\tau_{(j-1)})} \cdot C_{j-1} \right] \mathbbm{1}_{\mathcal{R}_j}(x).
\end{equation*}

Let us obtain the probability mass function of $X \mid X \geq \tau_{(k)}$ by definition:
\begin{align*}
    p_{X \mid X \geq \tau_{(k)}}(x) &= \dfrac{p(x)}{\mathbbm{P}(X \geq \tau_{(k)})} \cdot \mathbbm{1}_{[\tau_{(k)}, +\infty]}(x), \\
    &= \left[\frac{x^{-\alpha_{k+1}}}{\zeta(\alpha_{k+1},\tau_{(k)})} \cdot C_{k} \right] \left[\frac{\zeta(\alpha_{k+1},\tau_{(k)})}{\zeta(\alpha_{k+1},\tau_{(k)})} \cdot C_{k} \right]^{-1},\\
    &= \frac{x^{-\alpha_{k+1}}}{\zeta(\alpha_{k+1},\tau_{(k)})},\quad x \geq \tau_{(k)},
\end{align*}
which indeed is a power-law distribution with specified parameters.

\section{Simulating samples}\label{sec:samples}

%Pseudo-random samples from a discrete piecewise power-law distribution are generated using the discrete inverse transform method, as detailed in the following algorithm.

\begin{algorithm}
\caption{Simulating samples from a discrete piecewise power-law distribution.}
\begin{algorithmic}
\State \textbf{Input:} $n$ (number of samples), $\boldsymbol\tau$ (change-points vector), $\boldsymbol{\theta}$ (scaling parameters)
\State \textbf{Initialize:} Compute the cumulative distribution function $F(x)$
\For{ $i = 1$ to $n$ }
    \State Generate $u_i \sim \text{Uniform}(0,1)$ %\Comment{Uniform random variable}
    \State Set $x_i = \inf \{x_j \in \{1, 2, \ldots\} : F(x_j) \geq u_i \}$ %\Comment{Inverse transform sampling}
\EndFor
\State \Return $\{x_1, x_2, \dots, x_n\}$
\end{algorithmic}
\end{algorithm}

\clearpage
\section{Study 2}\label{appendix:study2}

\begin{figure}[!ht]
    \centering
    \includegraphics[scale = 0.37]{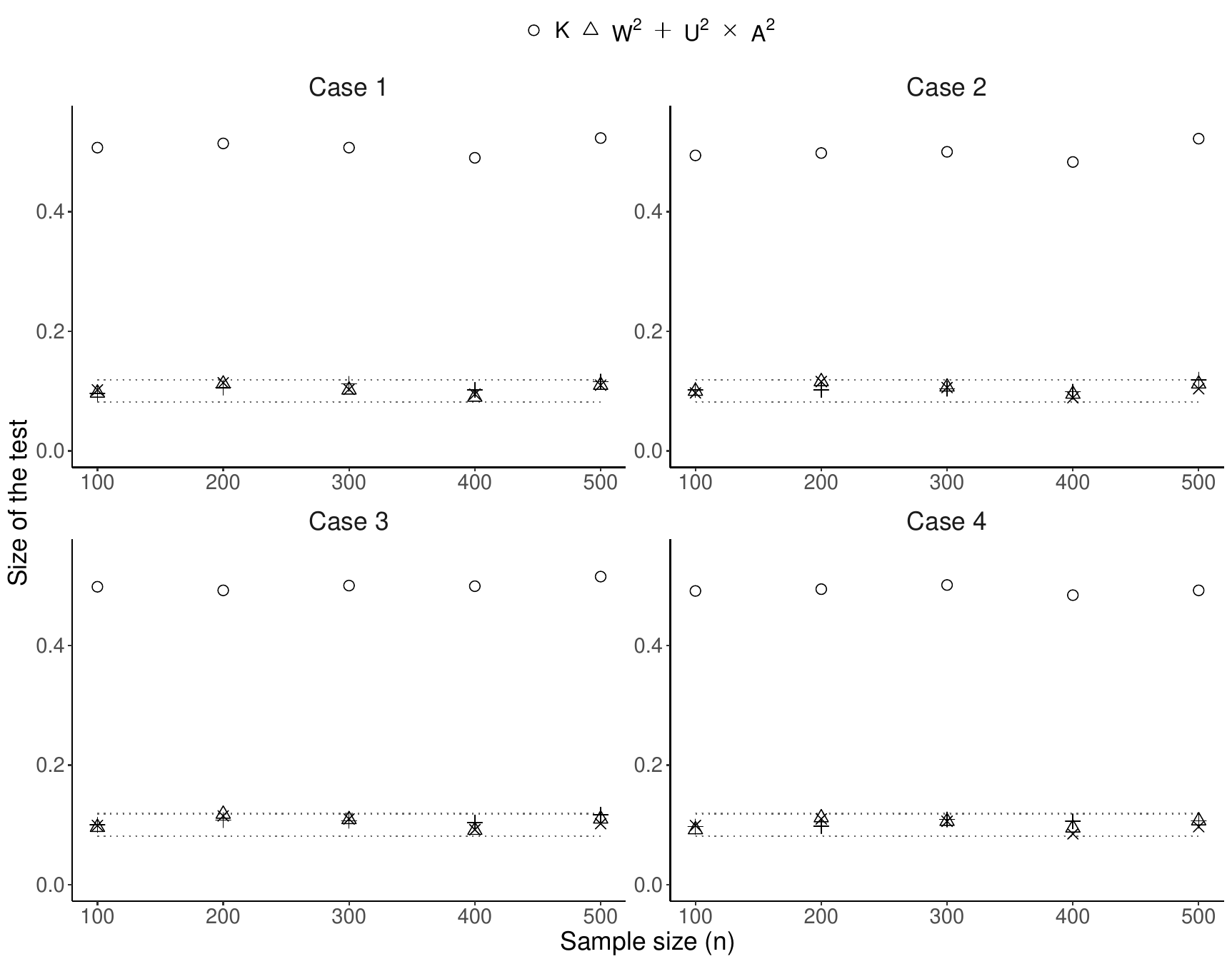}
    \caption{Empirical sizes of the GOF tests based on the $K$, $W^2$, $U^2$, and $A^2$ statistics at a 10\% significance level. The dashed line represents the limits of $0.1 \pm 1.96 \sqrt{0.1(1-0.1)/\text{1,000}}$, derived from the normal approximation.}\label{fig10}
\end{figure}

\begin{figure}[!ht]
    \centering
    \includegraphics[scale=0.37]{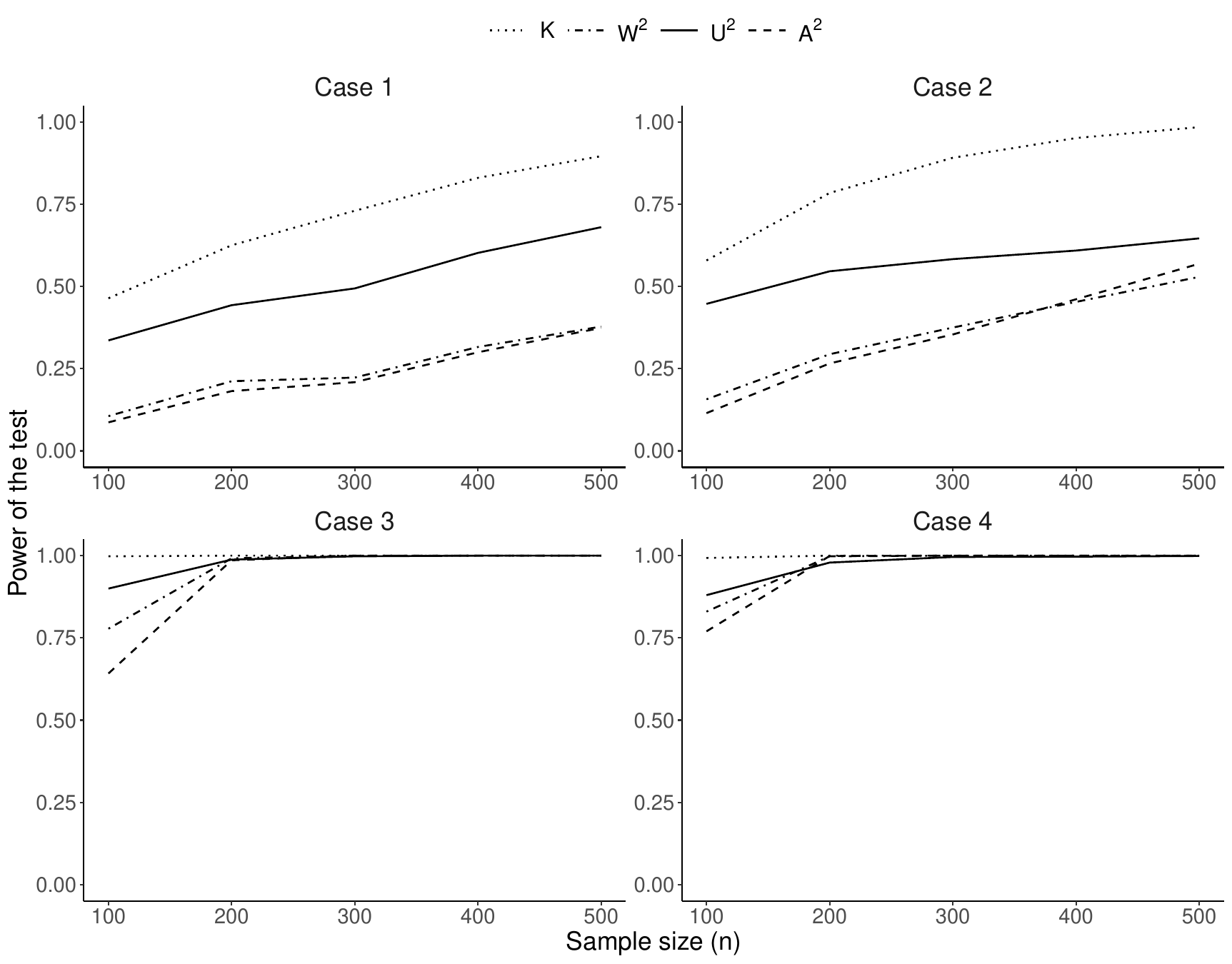}
    \caption{Empirical power of GOF tests based on the $K$, $W^2$, $U^2$, and $A^2$ statistics at a 10\% significance level.}\label{fig11}
\end{figure}

\clearpage
\section{Study 3}\label{appendix:study3}

\begin{table}[!ht]
\centering
\caption{Bias, MSE, and CP for the scaling parameter estimates of the discrete piecewise power-law model using the MLE and based on 1,000 simulated samples.}
\label{tab6}
\scalebox{0.8}{
\rotatebox{0}{
\begin{tabular}{ccccccccccccccc}
  \hline
  & & \multicolumn{3}{c}{$\hat\alpha_1$} & & \multicolumn{3}{c}{$\hat\alpha_2$} & & \multicolumn{3}{c}{$\hat\alpha_3$} \\
  \cline{3-5} \cline{7-9} \cline{11-13}
  & $n$ & Bias & MSE & CP & & Bias & MSE & CP & & Bias & MSE & CP \\ 
  \hline
  \multirow{5}{*}{Case 1:} & 100 & 0.00 & 0.00 & 0.94 &  & 0.06 & 0.14 & 0.95 & & - & - & - \\ 
   & 250 & 0.00 & 0.00 & 0.95 &  & 0.03 & 0.05 & 0.95 & & - & - & - \\ 
   & 500 & 0.00 & 0.00 & 0.96 &  & 0.01 & 0.03 & 0.95 & & - & - & - \\ 
   & 750 & -0.00 & 0.00 & 0.95 &  & 0.01 & 0.02 & 0.95 & & - & - & - \\ 
   & 1,000 & 0.00 & 0.00 & 0.95 &  & 0.00 & 0.01 & 0.95 & & - & - & - \\
   \hline
   \multirow{5}{*}{Case 2:} & 100 & 0.00 & 0.00 & 0.94 &  & 0.10 & 0.32 & 0.94 & & - & - & - \\ 
   & 250 & -0.00 & 0.00 & 0.96 &  & 0.02 & 0.11 & 0.94 & & - & - & - \\ 
   & 500 & 0.00 & 0.00 & 0.95 &  & 0.02 & 0.05 & 0.95 & & - & - & - \\ 
   & 750 & 0.00 & 0.00 & 0.95 &  & 0.02 & 0.03 & 0.95 & & - & - & - \\ 
   & 1,000 & 0.00 & 0.00 & 0.95 &  & 0.01 & 0.03 & 0.95 & & - & - & - \\ 
   \hline
   \multirow{5}{*}{Case 3:} & 100 & 0.00 & 0.00 & 0.95 &  & -0.01 & 0.04 & 0.90 &  & 0.22 & 0.44 & 0.92 \\ 
   & 250 & 0.00 & 0.00 & 0.96 &  & 0.00 & 0.01 & 0.94 &  & 0.06 & 0.10 & 0.94 \\ 
   & 500 & -0.00 & 0.00 & 0.95 &  & 0.00 & 0.01 & 0.95 &  & 0.02 & 0.04 & 0.94 \\ 
   & 750 & -0.00 & 0.00 & 0.95 &  & 0.00 & 0.00 & 0.96 &  & 0.01 & 0.03 & 0.96 \\ 
   & 1,000 & 0.00 & 0.00 & 0.96 &  & 0.00 & 0.00 & 0.95 &  & 0.01 & 0.02 & 0.94 \\ 
   \hline
   \multirow{5}{*}{Case 4:} & 100 & 0.00 & 0.00 & 0.93 &  & -0.03 & 0.05 & 0.92 &  & 0.21 & 0.65 & 0.91 \\ 
   & 250 & 0.00 & 0.00 & 0.94 &  & 0.01 & 0.02 & 0.94 &  & 0.09 & 0.18 & 0.94 \\ 
   & 500 & 0.00 & 0.00 & 0.96 &  & -0.00 & 0.01 & 0.94 &  & 0.03 & 0.07 & 0.95 \\ 
   & 750 & -0.00 & 0.00 & 0.95 &  & 0.00 & 0.00 & 0.95 &  & 0.02 & 0.05 & 0.94 \\ 
   & 1,000 & 0.00 & 0.00 & 0.94 &  & 0.00 & 0.00 & 0.95 &  & 0.00 & 0.03 & 0.95 \\ 
   \hline
\end{tabular}
}
}
\end{table}

\section{Study 4}\label{appendix:study4}

\begin{table}[!ht]
\centering
\caption{Empirical size and power (in parentheses) of the $K$, $W^2$, $U^2$, and $A^2$ statistics at a 10\% significance level for discrete piecewise power-law samples across different cases.}
\label{tab7}
\scalebox{0.8}{
\begin{tabular}{cccccccccccc}
  \hline
  & \multicolumn{4}{c}{Case 1} & & \multicolumn{4}{c}{Case 2} \\
  \cline{2-5} \cline{7-10}
  $n$ & $K$ & $W^2$ & $U^2$ & $A^2$ & & $K$ & $W^2$ & $U^2$ & $A^2$ \\
  \hline
   100 & 0.03 & 0.10 & 0.10 & 0.10 &  & 0.03 & 0.09 & 0.10 & 0.09 \\ 
   250 & 0.03 & 0.10 & 0.09 & 0.10 &  & 0.03 & 0.10 & 0.10 & 0.10 \\ 
   500 & 0.03 & 0.10 & 0.10 & 0.10 &  & 0.03 & 0.10 & 0.10 & 0.10 \\ 
   750 & 0.04 & 0.10 & 0.10 & 0.10 &  & 0.04 & 0.11 & 0.11 & 0.11 \\ 
   1,000 & 0.04 & 0.10 & 0.10 & 0.11 &  & 0.04 & 0.11 & 0.11 & 0.11 \\
   \hline
   & \multicolumn{4}{c}{Case 3} & & \multicolumn{4}{c}{Case 4} \\
   \cline{2-5} \cline{7-10}
   $n$ & $K$ & $W^2$ & $U^2$ & $A^2$ & & $K$ & $W^2$ & $U^2$ & $A^2$ \\
   \hline
    100 & 0.04 (0.07) & 0.09 (0.07) & 0.10 (0.49) & 0.09 (0.06) &  & 0.04 (0.04) & 0.09 (0.05) & 0.10 (0.42) & 0.09 (0.05) \\ 
    250 & 0.04 (0.70) & 0.10 (0.67) & 0.10 (0.96) & 0.10 (0.72) &  & 0.04 (0.47) & 0.10 (0.51) & 0.10 (0.92) & 0.10 (0.60) \\ 
     500 & 0.04 (1.00) & 0.10 (0.99) & 0.10 (1.00) & 0.10 (1.00) &  & 0.04 (0.96) & 0.10 (0.97) & 0.10 (1.00) & 0.10 (0.99) \\ 
     750 & 0.04 (1.00) & 0.11 (1.00) & 0.10 (1.00) & 0.10 (1.00) &  & 0.04 (1.00) & 0.11 (1.00) & 0.10 (1.00) & 0.11 (1.00) \\ 
     1,000 & 0.05 (1.00) & 0.11 (1.00) & 0.11 (1.00) & 0.11 (1.00) &  & 0.05 (1.00) & 0.11 (1.00) & 0.11 (1.00) & 0.11 (1.00) \\ 
   \hline
\end{tabular}}
\end{table}

\end{document}